\documentclass[reqno,11pt,centertags]{article}
\usepackage{amsmath,amsthm,amscd,amssymb,latexsym,upref}
\date{\today}

\input epsf
\usepackage{epsfig}
\usepackage[T2A,OT1]{fontenc}
\usepackage[ot2enc]{inputenc}
\usepackage[russian,english]{babel}

\usepackage{epic,eepicemu}

\newcommand{\bbD}{{\mathbb{D}}}
\newcommand{\cD}{{\mathcal{D}}}

\newcommand{\bbR}{{\mathbb{R}}}

\newcommand{\bbZ}{{\mathbb{Z}}}
\newcommand{\bbC}{{\mathbb{C}}}

\newcommand{\bbT}{{\mathbb{T}}}

\newcommand{\cF}{{\mathcal{F}}}

\newcommand{\cP}{{\mathcal{P}}}
\newcommand{\cK}{{\mathcal{K}}}
\newcommand{\cQ}{{\mathcal{Q}}}

\newcommand{\cE}{{\mathcal{E}}}

\newcommand{\fz}{{\mathfrak{z}}}
\newcommand{\e}{{\epsilon}}
\renewcommand{\l}{{\ell}}

\newcommand{\fA}{{\mathfrak{A}}}

\renewcommand{\k}{\varkappa}
\newcommand{\z}{\zeta}
\renewcommand{\Re}{\text{\rm Re\,}}
\newcommand{\Om}{{\bar\bbC\setminus E}}
\renewcommand{\Im}{\text{\rm Im\,}}

\newcommand{\sgn}{\text{\rm sgn}}
\renewcommand{\span}{\text{\rm span}}

\allowdisplaybreaks \numberwithin{equation}{section}
\newtheorem{theorem}{Theorem}[section]

\newtheorem{lemma}[theorem]{Lemma}
\newtheorem{proposition}[theorem]{Proposition}

\newtheorem{corollary}[theorem]{Corollary}

\theoremstyle{definition}
\newtheorem{definition}[theorem]{Definition}

\newtheorem{remark}[theorem]{Remark}

\begin{document}

\title{Kotani-Last problem and Hardy spaces on surfaces of Widom type}

\author{A. Volberg\thanks{A. V. is supported by NSF grant DMS 0758552, he also thanks the J. Kepler University of Linz for the hospitality during his visit in summer  2012.}\ \ and P. Yuditskii\thanks{In the frame of FWF project P25591-N25}}


%



\maketitle

\begin{abstract}
It is  a small theory of non almost periodic ergodic families of Jacobi matrices
with pure (however) absolutely continuous spectrum. And the reason why this effect may happen:
under our ``axioms" we found an  analytic condition on the resolvent set that is responsible for (exactly equivalent to)  this effect.

\medskip

\end{abstract}

\section{Introduction}

\subsection{Ergodic Jacobi matrices. Setting of the problem}
Latest developments in the spectral theory of Schr\"odinger operators, Jacobi and CMV matrices \cite{AD, AJ, GS, JS,  KS, KSDp, PR09, PR11, REMA11, KOT, KR}, and our own new level of understanding of the function theory in multi-connected domains \cite{Yu2011, Yu2012} (for a classical background see \cite{WID, Pom, Has}), provided a solid basis for solving
the Kotani-Last problem \cite{DD},  in fact, below we give a negative answer to their conjecture.
According to Arthur Avila   ``this problem has been around for a while, and became a central topic of the theory after recent popularization (by Barry Simon, Svetlana Jitomirskaya and David Damanik)".

In the setting of this problem we follow \cite{DD}. Suppose $(\Xi,d\xi)$ is a probability measure space, $T:\Xi\to\Xi$ is an invertible ergodic transformation, and $\log \cP:\Xi\to \bbR$ and $\cQ:\Xi\to \bbR$ are bounded and measurable. We define coefficient sequences 
$$
p_n(\xi)=\cP(T^n\xi), \ q_n(\xi)=\cQ(T^n\xi),\quad \xi\in \Xi, \ n\in\bbZ,
$$
and corresponding Jacobi matrices acting in $\l^2=\l^2(\bbZ)$
\begin{equation}\label{kli1}
J(\xi) e_n=p_n(\xi) e_{n-1}+q_n(\xi) e_n+p_{n+1}(\xi)e_{n+1},
\end{equation}
where $e_n$'s are the vectors of the standard basis. We call $\{J(\xi)\}_{\xi\in\Xi}$ an ergodic family of Jacobi matrices.

A non-random measure $d\omega(x)$, called the \textit{density of states measure}, is defined by the relation
$$
\int\frac{d\omega(x)}{x-z}=\int_{\Xi}\langle(J(\xi)-z)^{-1}e_0,e_0\rangle d\xi, \quad
z\in\bbC\setminus \bbR.
$$
  The almost sure spectrum of $J(\xi)$ is given by the closed  support of this measure (Avron-Simon, 1983).
  The integrated density of states and the \textit{Lyapunov exponent} are related by the Thouless formula
  \begin{equation}\label{thfo}
  G(z)=c+\int\log|z-x|d\omega(x), \quad c=-\int \log \cP(\xi) d\xi.
\end{equation}
Finally, the almost sure \textit{absolutely continuous spectrum} $E$ of $J(\xi)$ is given by the essential closure of the set of energies for which the Lyapunov exponent vanishes
  (Ishii, Pastur, Kotani), that is, by the essential closure of the set
  $
 \{z: G(z)=0\}.
  $

 The  \textit{Kotani-Last problem} asks to show that if the Lebesgue measure of $E$ is positive, $|E|>0$, then 
  $J(\xi)$ is almost periodic, that is, the set of translates $\{S^{-n}J(\xi)S^n\}_{n\in \bbZ}$, $Se_n=e_{n+1}$, is a pre-compact set in the operator topology.
  
  The best known example of an ergodic (and almost periodic) Jacobi matrix is the \textit{almost Mathieu operator}:
  $$
  \Xi=\bbT=\{e^{i\phi}\}, \ d\xi= dm\ \text{(the Lebesgue measure)}, \ Te^{i\phi}=e^{i(\phi+2\pi\tau)},
  $$
  and
  \begin{equation}\label{kotani1}
  \cP(e^{i\phi})=1, \quad \cQ(e^{i\phi})=2\lambda \cos\phi, \ \text{i.e.,}\ 
  q_n(e^{i\phi})=2\lambda\cos(2\pi n\tau+\phi)
\end{equation}
A \textit{general almost periodic Jacobi matrix} looks quite similar to this particular case, namely
  $
  \Xi= A$ is a compact Abelian group, $d\xi=d\alpha$ is the Haar measure on $A$, $T\alpha=\mu^{-1}\alpha$, $\mu\in A$, (multiplication by a fixed element of $A$), 
  and $\cP(\alpha)$ and $\cQ(\alpha)$ are arbitrary continuous functions on the group $A$.
  Let us point out that $T$ is ergodic if the trajectory $\{\mu^{-n}\alpha\}_{n\in\bbZ}$ is dense in $A$ ($\tau$ is irrational in the almost Mathieu case).
  
    The Kotani-Last conjecture was motivated by the following celebrated theorem.
  
  \begin{theorem}[Kotani, 1989]\label{thkotani1}
  Let $\cP(\xi)=1$ and $Q(\xi)$ take finitely many values. Assume that  $J(\xi)$ are $d\xi$-almost surely not periodic. Then the Lebesgue measure of $E$ is zero (i.e., the absolutely continuous spectrum is empty).
  \end{theorem}
  
 The spectrum of the almost Mathieu operator has been studied deeply. For example, three problems  of Barry Simon's fifteen problems about Schr\"o\-dinger operators ``for the twenty-first century", which featured the almost Mathieu operator \cite{Sim2000}, are now all solved. For $\lambda\in(0,1)$, $J(\phi)$ has almost surely purely absolutely continuous spectrum (it  has almost surely pure point spectrum for $\lambda>1$). But as soon as we substitute in \eqref{kotani1} the continuos function $2\lambda\cos\phi$ with $Q(e^{i\phi})=2\lambda\, \sgn(\phi)$, where $\sgn(\phi)=1$, $\phi\in(0,\pi)$ and $\sgn(\phi)=-1$ for $\phi\in(-\pi,0)$, then by the Kotani Theorem, the absolutely continuous  spectrum disappears for an arbitrary amplitude $\lambda>0$! So, the potential function $\cQ$ is continuous on $\bbT$ (the operator is almost periodic) and an a.c. spectrum is possible, but as soon as  the function $\cQ$ is discontinuous (the operator is not almost periodic) the a.c. spectrum is forbidden. Of course, we should have in mind that
 Theorem \ref{thkotani1} holds only for quite simple discontinuous functions (a finite number of jumps).
 
 A \textit{negative} answer to the Kotani-Last conjecture was given by  Arthur Avila \cite{A}. He made the following announcement on his website: 
  \begin{itemize}
 \item[] \textit{Absolute continuity without almost periodicity}:
 ``We produce a counterexample which is like a ``limit-periodic driven" almost Mathieu operator. To guarantee the existence of some absolutely continuous spectrum, we work in the perturbative regime and do a bit of KAM". 
  \end{itemize}
  In the beginning of October, 2012 preprint \cite{A} appeared, we were made aware of that  after writing this article and sending it to several colleagues.
 
Naturally, it is  important and highly challenging not only  to prove or disprove the conjecture, but   also to understand the reasons for such an interesting phenomenon.
There are at least two programs dealing with this problematic: 
 \begin{itemize}
 \item[]
 S. Kotani, \textit{Grassmann manifold and spectral theory of 1-D Schr\"odinger operators}, Mini courses and Workshop "Hilbert spaces of entire functions and spectral theory of self-adjoint differential operators", at CRM and IMUB, 2011, 
 \end{itemize}
 see also \cite{KOT}. Another one is being  developed by C. Remling (or Poltoratski and Remling) \cite{PR09, PR11, REM11, REMA11}.

In this paper we present our approach (very different from Avila's approach in \cite{A}) to giving a negative answer to the Kotani--Last question. Our approach is dual to that of A. Avila in the following sense. Starting from a certain operator he proves  its required spectral property. We suggest to use methods of the
 \textit{inverse spectral theory}. That is, starting with a certain class of Jacobi matrices given by means of their spectral properties we will prove the required properties of the generated coefficient sequences. In \cite{A} the discrete setting represents considerably higher level of difficulties in comparison with a continuous case. In particular, in the discrete setting in \cite{A} the spectrum is not pure absolutely continuous.  Unlike \cite{A}, our ergodic family has only pure absolutely continuous spectrum, but, on the other hand, we have Jacobi matrices and not discrete Schr\"odinger operators ($\mathcal{P}=1$ in \cite{A}). See also remark \ref{ravno1} below. It is interesting to note that in our case a ``simple" analytic property of the resolvent domain is totally responsible for almost periodicity.
 See remark \ref{switch} below.

By Kotani theory \cite{DD} ergodic matrices with a.c. spectrum are \textit{reflectionless} and we start our construction with the definition of this class.

\subsection{Reflectionless Jacobi matrices: class $J(E)$. Main result}

Let  $J$ be a two-sided matrix acting in $\l^2$
\begin{equation}\label{kli1bis}
Je_n=p_n e_{n-1}+q_n e_n+p_{n+1}e_{n+1}.
\end{equation}
Let 
$\l^2_+=\span_{n\ge 0}\{e_n\}$ and $\l^2_-=\span_{n\le -1}\{e_n\}$. We define
\begin{equation*}\label{kli2}
J_{\pm}= P_{\l^2_\pm} J|\l_\pm^2,
\end{equation*}
and corresponding resolvent functions
 \begin{equation}\label{kli3}
r_{\pm}(z)=\langle(J_\pm -z)^{-1} e_{\frac{-1\pm 1}2},  e_{\frac{-1\pm 1}2}\rangle.
\end{equation}

\begin{definition}
For a compact $E=[b_0,a_0]\setminus \cup_{j\ge 1}(a_j,b_j)$ of  positive Lebesgue measure, $|E|>0$, we say that $J\in J(E)$ if the spectrum of $J$ is $E$ and
 \begin{equation}\label{kli4}
\frac 1{r_{+}(x+i0)}=\overline{p_0^2 r_-(x+i0)}\ \text{for almost all}\ x\in E.
\end{equation}
\end{definition}

Notice that $r_{\pm}(z)$ are functions with positive imaginary part in the upper half-plane (in what follows we say functions of Nevanlinna class). As it well known such functions possess boundary limit values for almost all $x\in\bbR$. 

The property \eqref{kli4} is shift invariant, that is, $S^{-1}JS\in J(E)$ as soon as $J\in J(E)$.
$J(E)$ is a compact in the sense of pointwise convergence of coefficients sequences.

We are specially interested in the case when every element of $J(E)$ has \textit{absolutely continuous} spectrum. For instance, if $E$ is a system of non-degenerated intervals with a unique accumulation point, say $0$, then $E$ is a subject for the condition
$$
\int_E\frac{dx}{|x|}=\infty.
$$
In general the structure of corresponding sets $E$ was  studied  quite completely by Poltoratski and Remling \cite{PR09}.

To be able to use freely Complex Analysis we assume, in addition, that the resolvent domain $\Omega=\bar \bbC\setminus E$ is of Widom type.
One can give a \textit{parametric description of regular Widom domains} by means of conformal mappings on the so called comb-domains, see survey \cite{EY}.

Consider two sequences $\{\omega_k\}$, $\omega_k\in (0,1)$, $\omega_k\not=\omega_j$ if
$k\not=j$, and $\{h_k\}$, $h_k>0$, $\sum h_k<\infty$. Let $\Pi$ be a region obtain from the half-strip
$$
\{z=x+iy:\ -\pi <x<0,\ y>0\}
$$
by removing vertical intervals $\{-\pi\omega_k+iy: 0<y\le h_k\}$. Let $\theta$ be the conformal map from the upper half-plane to $\Pi$ such that
$$
\theta(b_0)=-\pi, \ \theta(a_0)=0,\ \theta(\infty)=\infty.
$$
Notice that it has a continuous extension to the closed half-plane. The set $E$ corresponds to the base of the comb, $E=\theta^{-1}([-\pi,0])$, and the gaps correspond to the slits.

By the symmetry principle $\theta(z)$ can be extended to $\Omega$ as \textit{multivalued function},
but $G(z):=\Im\theta(z)$ is single valued harmonic function with logarithmic pole at infinity, such that 
$G(x\pm i0)=0$ for $x\in E$. In other words, $G(z)=G(z,\infty)$ is the Green function in $\Omega$ with respect to infinity. 

Further, the Lebesgue measure on the base of the comb corresponds to the \textit{harmonic measure} $d\omega$ on $E$ and we have \eqref{thfo}. In particularly we get immediately that $d\omega(x)$ and $dx$ are mutually absolutely continuous on $E$.

To work with multivalued functions,  it is convenient to introduce from the very beginning a universal covering of $\Omega$. Let $\bbD/\Gamma\simeq \bar \bbC\setminus E$, that is, there exists a Fuchsian group $\Gamma$ and a meromorphic function $\fz:\bbD\to \bar \bbC\setminus E$, $\fz\circ\gamma=\fz$ for all $\gamma\in\Gamma$, such that
$$
\forall z\in\Om\  \exists \z\in\bbD: \fz(\z)=z\ \text{and}\ \fz(\z_1)=\fz(\z_2)\Rightarrow \z_1=\gamma(\z_2).
$$
We assume that $\fz$ meets the normalization conditions $\fz(0)=\infty$, $(\z\fz)(0)>0$. 

Of course $\Gamma$ is equivalent to the fundamental group of the domain $\pi_1(\Omega)$ and generators $\gamma_j$'s of $\Gamma$ corresponds to the closed curves that start on the positive half-axis $x>a_0$, go in the upper half plane to the gap $(a_j,b_j)$'s and go back in the lower half-plane. 

For instance, in this case for the function $b(\zeta):=e^{i\theta(\fz(\z))}$ we have
$$
  b\circ\gamma_j=\mu(\gamma_j) b,
\ \text{where}\ \mu(\gamma_j)=e^{-2\pi i\omega_j}.
$$
The system of multipliers $\{\mu(\gamma)\in\bbT\}_{\gamma\in\Gamma}$ form a character of the group $\Gamma$, i.e.,
$\mu(\gamma_1\gamma_2)=\mu(\gamma_1)\mu(\gamma_2)$. Thus $b(\z)$ is so called \textit{character automorphic function} in $\bbD$ with respect to $\Gamma$.

Now, every meromorphic automorphic function $f$ in $\bbD$, $f\circ\gamma=f$,  is of the form $f=F\circ\fz$, where $F$ is meromorphic in $\Om$ and vice versa.
 We say that $F(z)$ belongs to the Smirnov class $N_+(\Omega)$ if $f=F\circ\fz$ is of Smirnov class in $\bbD$, i.e., it can be represented as a ratio of two bounded functions $f=f_1/f_2$ and the denominator $f_2$ is an outer function.

Note that the Lebesgue measure on $\bbT$ corresponds to the harmonic measure on $E$. Thus
$F(z)\in N_+(\Omega)$ has boundary values for almost all $x\in E$ with respect to the Lebesgue measure $dx$, restricted to this set.

Now we are ready to formulate our main result, but we need to define one more property 
of domains of Widom type.

\begin{definition} We say that the Direct Cauchy Theorem (DCT) holds in a Widom domain $\Omega$ if
\begin{equation}\label{DCT1}
\frac 1{2\pi i}\oint_{\partial \Omega}\frac{F(x)}{x-z}dx=F(z), \quad z\in\Omega,
\end{equation}
for all $F\in N_+(\Omega)$ such that 
\begin{equation}\label{DCT2}
\oint_{\partial \Omega}|{F(x)}dx|<\infty \ \text{and}\ F(\infty)=0.
\end{equation}
\end{definition}

\begin{theorem}\label{thmain}
Let $E$ be such that
\begin{itemize}
\item[(i)] Every  $J\in J(E)$ has  absolutely continuous spectrum;
\item[(ii)] The domain $\Omega=\Om$ is of Widom type;
\item[(iii)] The frequencies   $\{\omega_j\}$ are independent in the sense that for every finite collection of indexes 
$I$ and $n_j\in\bbZ$
$$
\sum_{j\in I}\omega_j n_j=0\ \text{\rm mod} \ 1
$$
implies that  $n_j=0$ for all $j\in I$.
\end{itemize}
Then on the compact $\Xi:=J(E)$ there is an  unique shift invariant measure $dJ$, $dJ=d(S^{-1}J S)$, such that the measurable functions, see \eqref{kli1bis},
\begin{equation}\label{functionstut}
\cP(J)=p_0=\langle J e_0,e_{-1}\rangle,\quad \cQ(J)=q_0=\langle J e_0,e_{0}\rangle, \quad J\in J(E),
\end{equation}
define an ergodic family $\{J\}_{J\in J(E)}$.

Moreover, all elements of the family are either almost periodic or they all are not, depending on the following analytic property of the domain $\Omega$: either
DCT holds true or it fails.
\end{theorem}

\begin{remark}
The statement that DCT implies that all elements of $J(E)$ are almost periodic was proved in \cite{SY}. Note that DCT implies automatically  that all $J$'s have absolutely continuous spectrum, and  the condition $(iii)$ is also redundant. 
\end{remark}

Our proof of the main part of Theorem \ref{thmain} is as follows. Widom condition allows us to construct \textit{generalized Abel map} $\pi$ from $J(E)$ to the group of characters $\Gamma^*$. As soon as DCT fails this map is not one to one. It deals with a hidden singularity in Hardy spaces of character automorphic functions  on $\Omega$, which exists  even in the case that all $J$'s have no singular  spectral components. Nevertheless, based on the main analytic lemma, we show that this map is one to one almost everywhere with respect to the Haar measure on $\Gamma^*$ and therefore it defines a  unique shift invariant measure on $J(E)$. On the other hand, presence of different Jacobi matrices in $J(E)$, corresponding to the same character implies that they are not almost periodic.

\begin{remark}
Concerning almost periodicity: the functions $\cP(J)$ and $\cQ(J)$ defined by \eqref{functionstut} are continuous on the compact $J(E)$, but this set does not have the structure of an Abelian group as soon as DCT fails. On the other hand essentially the same ergodic family can be given by means of the functions $\hat \cP(\alpha)$ and $\hat \cQ(\alpha)$ on $\Gamma^*$, see Theorem \ref{th44}. Now $\Gamma^*$  is an Abelian group, but the functions   $\hat \cP(\alpha)$ and $\hat \cQ(\alpha)$ are not continuous as  soon as DCT fails, of course this discontinuity is of a very non-trivial structure.
\end{remark}

\begin{remark}
We mentioned above a hidden singularity in Hardy spaces of character automorphic functions  on $\Omega$, which exists  even in the case that all $J$'s have no singular  spectral components. We describe it in Section 5 by means of the so called $j$-inner matrix functions. Revealing of this structure allows us easily construct \textit{examples} of Widom Domains of this sort \textit{without DCT},  see \cite{Yu2011}. Note, by the way, that the main result of \cite{SY} is stated for the so called homogeneous sets $E$. If $E$ is homogeneous then DCT holds in $\Omega=\Om$. For a certain period there was a conjecture that the homogeneity is also a necessary condition for DCT. However in \cite{Yu2011} examples of non-homogeneous boundaries of  Widom domains with DCT are given.
\end{remark}

\begin{remark}
\label{ravno1}
The special case of Jacobi matrices, when $\cP(\xi)=1$, the so called discrete Schr\"odinger operators, is studied very intensively, see e.g. \cite{AD, AJ, GS, JS}, and in \cite{A} the Kotani-Last conjecture is studied in this setting.
Recall that we are working in the framework of the inverse spectral theory. From this point of view  $\cP(\xi)=1$ is an a priori extra condition and  this is a long-standing problem, which was formulated explicitly for instance by V.A. Marchenko: ``clarify which spectral data correspond exactly to discrete Schr\"odinger operators".
We did not investigate this problem in our context, in particular, 
it is not clear now whether [1] can be derived from our result.
\end{remark}

\begin{remark}
\label{switch}
Finally, we restate our main result in the following words: for an ergodic Jacobi matrices with the \textit{purely} absolutely continuos spectrum in a \textit{generic} situation (under a  probably technical Widom condition) we found a kind of switch with only two positions: DCT holds or it fails. As soon as the switch is in the position ``on" all operators of the class are almost periodic as soon as it is ``off"  none of them is. Amazingly,  an \textit{analytic} property of the resolvent domain completely characterizes  \textit{almost periodicity}.
\end{remark}

\section{Hardy spaces}

\subsection{Widom's function}
We described Widom domains by means conformal mappings $\theta(z)$. Recall that
$G(z)=G(z,\infty)=\Im\theta(z)$ is the Green function with respect to infinity.  Denote $c_j=\theta^{-1}(-\pi 
\omega_j+ih_j)$, $c_j\in(a_j,b_j)$. They are critical points of $G(z)$. Thus $\sum h_j<\infty$ is in fact the Widom condition
\begin{equation}\label{hs1}
\sum_{c_j:\nabla G(c_j)=0}G(c_j)<\infty.
\end{equation}

Fix $z_0\in\Omega$ and let orb$(\z_0)=\fz^{-1}(z_0)=\{\gamma(\z_0)\}_{\gamma\in\Gamma}$. The Blaschke product $b_{z_0}$ with zeros $\fz^{-1}(z_0)$ in the disc is related to the Green function $G(z,z_0)$ by
\begin{equation}\label{hs2}
\log\frac{1}{|b_{z_0}(\z)|}=G(\fz(\z), z_0).
\end{equation}

We normalize $b_{z_0}$ as follows: $b_{z_0}(0)>0$. It is a character automorphic function with character $\mu_{z_0}$. 

We put
\begin{equation}\label{hs3}
\Delta(\z):=\prod_{j\ge 1} b_{c_j}(\z).
\end{equation}
Then $\Delta(0)=e^{-\sum_{j\ge 1} h_j}$ by \eqref{hs2}. We use notation $\nu$ for the character of $\Delta$, $\Delta\circ\gamma=\nu(\gamma)\Delta$.

Let $\Gamma^*$ be the group of characters, it is provided  with topology of convergence on each element of the discrete group $\Gamma$. 

For $\alpha\in\Gamma^*$ we define
$$
H^\infty(\alpha)=\{f\in H^\infty: \ f\circ\gamma=\alpha(\gamma) f\},
$$ 
where $H^\infty$ means the standard Hardy space in $\bbD$.

According to the fundamental Widom theorem
\begin{equation*}\label{hs4}
\Delta(0)=\inf_{\alpha\in\Gamma^*}\sup_{f\in H^{\infty}(\alpha):\|f\|\le 1} |f(0)|.
\end{equation*}
Thus the Widom condition \eqref{hs1} is equivalent to the fact that $H^\infty(\alpha)$ contains a non-constant element for all $\alpha\in\Gamma^*$.

\subsection{The Hardy space $\hat H^p(\alpha)$}
Similarly to $H^\infty$-spaces for $p\ge 1$ we define
$$
\hat H^p(\alpha)=\{f\in H^p: \ f\circ\gamma=\alpha(\gamma) f\}, \ \alpha\in\Gamma^*.
$$
Then we introduce $F(z)$: 
\begin{equation}\label{hs005}
F\circ\fz=f\in\hat H^p(\alpha). 
\end{equation}
We can easily see that it is an modular automorphic function, that is, $|F|\circ\gamma=F$ for all $\gamma\in\pi_1(\Omega)$, and moreover 
\begin{equation}\label{hs5}
F\circ\gamma=\alpha(\gamma) F, \ \forall\gamma\in\pi_1(\Omega).
\end{equation}
$F$ is obviously in Smirnov's class of $\Omega$ and
\begin{equation}\label{hs6}
\int_E |F(x)|^p d\omega(x)=\int_\bbT |f(\z)|^p dm(\z)<\infty,
\end{equation}
where $d\omega$ is the harmonic measure on $E$ and $dm$ is the Lebesgue measure on $\bbT$.

Conversely,  any character automorphic function $F$ satisfying \eqref{hs5} with the finite norm \eqref{hs6}, which is in Smirnov's class in $\Omega$ can be obtained by formula \eqref{hs005}.

The latter space of functions will be called $\hat H^p_{\Omega}(\alpha)$. We saw that
\begin{equation}\label{hs103}
\hat H^p_{\Omega}(\alpha)\simeq\hat H^p(\alpha)
\end{equation}

by the operator $F\to F\circ\fz$.

\subsection{The Hardy space $\check H^q(\alpha)$. \\ Direct and Inverse Cauchy Theorems}

This family of spaces will be defined via orthogonal complements. 

For the standard duality between the spaces $L^p$ and $L^q$, $1/p+1/q=1$,  we use the notation
$$
\langle f,g\rangle =\int_\bbT \bar g f dm.
$$
An important role will be played by the following averaging operator \cite{Pom}, originally defined on functions $f\in L^\infty$,
\begin{equation}\label{hs7}
\displaystyle
P^\alpha f=\frac{\sum_{\gamma\in\Gamma} \alpha^{-1}(\gamma) |\gamma'| f\circ\gamma}
{\sum_{\gamma\in\Gamma} |\gamma'|}\in L^\infty(\alpha).
\end{equation}
Recall that the Green's function $G(z,\infty)$ corresponds to the Blaschke product $b=b_\infty$. 
Since $|b'(\z)|=\sum_{\gamma\in\Gamma}|\gamma'(\z)|$, $\z\in\bbT$, the expression \eqref{hs7}
can be rewritten as follows
\begin{equation}\label{hs8}
\displaystyle
P^\alpha f=\frac{b}{b'} \sum_{\gamma\in\Gamma} \alpha^{-1}(\gamma) \frac{\gamma'}{\gamma} f\circ\gamma.
\end{equation}

The main properties of $P^\alpha$ are:
\begin{itemize}
\item $\langle P^\alpha f, g\rangle=\langle  f, g\rangle$ for all $f\in L^\infty$ and $g\in L^1(\alpha)$,
\item $\| P^{\alpha}f\|_p\le \|f\|_p$.
\end{itemize}

Notice that the expression for $P^\alpha$ \eqref{hs7} and its main properties   follow mainly from the fact: for a Widom domain $\Omega\simeq\bbD/\Gamma$ there exists a measurable fundamental set for the action of the corresponding group $\Gamma$ on the unit circle $\bbT$ \cite{Pom}.

Turn now to the Hardy spaces. The representation \eqref{hs8} implies that $P^{\alpha}$ maps
$\Delta H^\infty$ to $H^\infty(\alpha)$. It follows from the fact that for Widom domains $b'$ is of Smirnov's class, moreover its inner part is the Widom function $\Delta$ (the Blaschke product over critical points) \cite{Pom}. Evidently,
\begin{equation}\label{hs9}
(P^\alpha\Delta h)(0)=\Delta(0) h(0), \ \forall h\in H^\infty.
\end{equation}

The following statement is known as the Inverse Cauchy Theorem \cite{Has}.

\begin{theorem}\label{thICT}
The annihilator of $\hat H_0^{p}(\alpha)=\{f\in \hat H_0^{p}(\alpha): f(0)=0\}$, $1\le p<\infty$, is contained in the subspace
\begin{equation}\label{hs10}
\Delta \overline{\hat H^q(\nu\alpha^{-1})}=
\{ \Delta \bar g: \ g\in \hat H^q(\nu\alpha^{-1})\}\subset L^q(\alpha).
\end{equation}
\end{theorem}
\begin{proof}
Let $f\in L^q(\alpha)$ and $\langle g, f\rangle=0$ for all $g\in \hat H_0^{p}(\alpha)$. Then for any function $h\in H^\infty$, $h(0)=0$, we have
\begin{equation}\label{00hs11}
0=\langle P^{\alpha}\Delta h, f\rangle=\langle \Delta h, f\rangle=\langle  h, \bar \Delta f\rangle.
\end{equation}
Hence, $\Delta \bar f\in H^q\cap L^{q}(\alpha^{-1}\nu)=\hat H_0^q(\nu\alpha^{-1})$.
\end{proof}
 
\begin{definition}
The space $\check H^q(\alpha^{-1}\nu)$ consists of all functions $g$ such that $f=\Delta\bar g$ belongs to the annihilator of $\hat H_0^{p}(\alpha)$.
\end{definition}

By Theorem \ref{thICT}, $\check H^q(\alpha^{-1}\nu)$ is a closed subspace of  $\hat H^q(\alpha^{-1}\nu)$. Such spaces coincide for all $\alpha$ if and only if the Direct Cauchy Theorem holds, that is,
\begin{equation}\label{hs11}
\int_\bbT\frac{f}{\Delta}dm=\langle  f,  \Delta \rangle=
\left(\frac{f}{\Delta}\right)(0), \ \forall f\in \hat H^1(\nu).
\end{equation}
Indeed, it remains to show that $\Delta\bar g$ annihilates $\hat H_0^p(\alpha)$ as soon as 
$g\in \hat H^q(\alpha^{-1}\nu)$. But, by \eqref{hs11}
$$
\langle  f,  \Delta\bar g \rangle=\langle  f g,  \Delta \rangle=0
$$
since $fg\in\hat H_0^1(\nu)$.

Finally notice that \eqref{00hs11}, in fact, implies that
\begin{equation}\label{hs12}
\check H^p(\alpha)=\text{clos}_{L^p}P^\alpha(\Delta H^\infty).
\end{equation}

\subsection{Reproducing kernels in $\hat H^2(\alpha)$ and $\check H^2(\alpha)$}

We define $\hat k^\alpha_{\zeta_0}(\z)=\hat k^\alpha(\zeta,\z_0)$ and 
$\check k^\alpha_{\zeta_0}(\z)=\check k^\alpha(\zeta,\z_0)$ by
\begin{equation}\label{hs13}
\begin{matrix}
\langle f, \hat k^\alpha_{\zeta_0}\rangle = f(\z_0), & \forall f\in\hat H^2(\alpha)\\
\langle g, \check k^\alpha_{\zeta_0}\rangle = g(\z_0), & \forall g\in\check H^2(\alpha)
\end{matrix}
\end{equation}
$\hat k^{\alpha}$, $\check k^\alpha$ mean $\hat k_0^{\alpha}$, $\check k_0^\alpha$
correspondingly.

\begin{lemma} For an arbitrary $\alpha\in\Gamma^*$
\begin{equation}\label{00hs14}
\begin{matrix}
L^2(\alpha)&=&\Delta \overline{\check H_0^2(\alpha^{-1}\nu)}
&\oplus& \{\hat e^{\alpha}\}
&\oplus &\hat H_0^2(\alpha)\\
&=&\Delta \overline{\check H_0^2(\alpha^{-1}\nu)}&\oplus &\{\Delta\overline{\check e^{\alpha^{-1}\nu}}\}
&\oplus& \hat H_0^2(\alpha)
\end{matrix}
\end{equation}
where
$$
\hat e^\alpha(\z)=\frac{\hat k^{\alpha}(\z)}{\|\hat k^\alpha\|},
\quad \check e^\alpha(\z)=\frac{\check k^{\alpha}(\z)}{\|\check k^\alpha\|}.
$$
Moreover, 
\begin{equation}\label{hs14}
\hat e^\alpha=\Delta\overline{\check e^{\alpha^{-1}\nu}} \ \text{on}\ \bbT 
\end{equation}
and 
\begin{equation}\label{bishs14}
\sqrt{\hat k^{\alpha}(0)
\check k^{\alpha^{-1}\nu}(0)}=\hat e^\alpha(0)\check e^{\alpha^{-1}\nu}(0)=\Delta(0).
\end{equation}
\end{lemma}
\begin{remark}
Equality \eqref{00hs14} means that the annihilator of $\hat H^2(\alpha)$ is exactly $\Delta \overline{\check H^2_0 (\alpha^{-1}\nu)}$.
\end{remark}
\begin{proof}
Notice that $\hat k^{\alpha}$ is orthogonal to $\hat H_0^2(\alpha)$, and therefore belongs to 
$\Delta\overline{\check H^2(\alpha^{-1}\nu)}$. On the other hand, on the dense set in 
$\check H^2(\alpha^{-1}\nu)$ we have
$$
\langle P^{\alpha^{-1}\nu}\Delta h, \Delta \overline{\hat k^{\alpha}}\rangle=
\langle  h, \overline{\hat k^{\alpha}}\rangle=\int_\bbT h \hat k^\alpha dm=h(0)\hat k^\alpha(0)=
\left(\frac{\hat k^\alpha}{\Delta}\right)(0) (\Delta h)(0).
$$
Thus  $\Delta \overline{\hat k^{\alpha}}$ is collinear to the reproducing kernel in $\check H^2(\alpha^{-1}\nu)$ and so \eqref{hs14} holds. But \eqref{hs14} implies obviously \eqref{bishs14}.
Evidently,
$\hat H^2(\alpha)=\{\hat e^{\alpha}\}
\oplus \hat H_0^2(\alpha)$ and similar for the check-spaces, so \eqref{00hs14} is also proved.
\end{proof}
\begin{corollary} The reproducing kernels in origin are uniformly bounded:
\begin{equation}\label{01hs14}
\Delta^2(0)\le \check k^\alpha(0)\le \hat k^{\alpha}(0)\le 1.
\end{equation}
\end{corollary}
\begin{proof}
The last inequality holds since $\hat k^{\alpha}$ is the projection of $1$ on $\hat H^{2}(\alpha)$.
The first one follows from \eqref{bishs14}.

\end{proof}

\begin{corollary}\label{cor33} We have
$$
\hat H_0^2(\alpha)=b\hat H^2(\alpha\mu^{-1}), \quad 
\check H_0^2(\alpha)=b\check H^2(\alpha\mu^{-1}).
$$
Also,
$$
\fz\hat H_0^2(\alpha)\subset \hat H^2(\alpha), \quad 
\fz\check H_0^2(\alpha)\subset\check H^2(\alpha).
$$
\end{corollary}
\begin{proof}
It is evident for hat-spaces. Therefore it holds also for check-spaces by \eqref{00hs14}.
Moreover, for $f\in \hat H_0^2(\alpha)$
\begin{equation}\label{hs15}
\fz f=(\fz f)(0)\frac{\hat k^{\alpha}}{\hat k^{\alpha}(0)}+g,
\end{equation}
where $g\in \hat H_0^2(\alpha)$, and similarly for $f\in  \check H_0^2(\alpha)$.

\end{proof}

\subsection{Intermediate Hardy space $H^2(\alpha)$ and \\ related Jacobi matrix}\label{subsec25}

\begin{definition}\label{def41}
By $H^2(\alpha)$ we denote an arbitrary close space such that
\begin{equation}\label{hs16}
\check H^2(\alpha)\subseteq H^2(\alpha)\subseteq \hat H^2(\alpha)
\end{equation}
and
\begin{equation}\label{hs17}
\fz H_0^2(\alpha)\subseteq H^2(\alpha), \ \text{where}\  
H_0^2(\alpha):=\{f\in H^2(\alpha): f(0)=0\}.
\end{equation}
\end{definition}

Let $k^\alpha(\z)$ be the reproducing kernel of this space.
We have an evident decomposition
\begin{equation}\label{hs18}
H^2(\alpha)=\{e^\alpha\}\oplus
H_0^2(\alpha), \ e^{\alpha}(\z):=\frac{k^\alpha(\z)}{\sqrt{k^\alpha(0)}}
\end{equation}

Note that every function of $H^2_0(\alpha)$ is of the form $f=bg$, $g\in \hat H^2(\alpha\mu^{-1})$.
We define $H^2(\alpha\mu^{-1})$ by $H_0^2(\alpha)=:b H^2(\alpha\mu^{-1})$. We have to check that this definition does not conflict with Definition \ref{def41}. But, indeed,
$\check H_0^2(\alpha)\subseteq H_0^2(\alpha)\subseteq \hat H_0^2(\alpha)$ and a function $f$  from $b H_0^2(\alpha\mu^{-1})$ has double zero in the origin. Therefore $\fz f$ belongs to $H^2(\alpha)$ and $(\fz f)(0)=0$. Thus both conditions \eqref{hs16} and \eqref{hs17} are fulfilled.

Thus \eqref{hs18} can be rewritten as
\begin{equation*}\label{hs19}
H^2(\alpha)=\{e^\alpha\}\oplus
b H^2(\alpha\mu^{-1}),
\end{equation*}
and we can iterate this relation
\begin{equation*}\label{hs20}
H^2(\alpha)=\{e^\alpha\}\oplus
\{be^{\alpha\mu^{-1}}\}\oplus
b^2 H^2(\alpha\mu^{-2}),
\end{equation*}
and so on...

\begin{lemma}\label{l42}
The system of vectors 
\begin{equation}\label{hs21}
e_n(\z)=e_n^\alpha(\z):=b^n \frac{k^{\alpha\mu^{-n}}(\z)}{\sqrt{k^{\alpha\mu^{-n}}(0)}}, 
\   n\in \bbZ_+,
\end{equation}
forms an orthonormal basis in $H^2(\alpha)$.
\end{lemma}
\begin{proof}
Note that a function $f\in H^2(\alpha)$, which is orthogonal to all vectors of the system has zero of infinite multiplicity in the origin. Thus $f=0$.
\end{proof}

Now we will shift $H^2(\alpha)$ in the "negative direction". Note that $e_0(0)>\check e^\alpha(0)>\Delta (0)$, see \eqref{01hs14}. Therefore $\fz e_0$ does not belong to $H^2(\alpha)$ (it has a simple pole at the origin). We define $e_{-1}(\z)$ via its orthogonal projection
\begin{equation}\label{00hs21}
\fz e_0=p_0e_{-1}+q_0e_0+p_1 e_1,\quad \|e_{-1}\|=1,\  (be_{-1})(0)>0.
\end{equation}
Here we took into account that $\fz e_0$ is orthogonal to $b^2H^2(\alpha\mu^{-2})$. Thus $e_{-1}$ is orthogonal to all $e_n$, $n\ge 0$.

We define $H^2(\alpha\mu)$ by
\begin{equation}\label{hs22}
b^{-1}H^2(\alpha\mu)=\{e_{-1}\}\oplus H^2(\alpha).
\end{equation}
Again we have to check that this definition is correct. Evidently $H^2(\alpha\mu)\subseteq \hat H^2(\alpha\mu)$. $H_0^2(\alpha\mu)=bH^2(\alpha)$, so $\fz H_0^2(\alpha\mu)\subset H^2(\alpha\mu)$ by definition of $e_{-1}$. Finally, since $\fz \check k^{\alpha}\in \fz H^{2}(\alpha)$,
by Corollary \ref{cor33}, we have the inclusion  $\check H^2(\alpha\mu)\subseteq  H^2(\alpha\mu)$.

\begin{remark}
Let us note that
\begin{equation}\label{00hs22}
be_{-1}(\z)(be_{-1})(0)=k^{\alpha\mu}(\z)
\end{equation}
Indeed, if
$$
f=c_0 (b e_{-1})+c_1{be_0}+c_2(b e_1)+\dots\in H^2(\alpha\mu)
$$
then
$$
f(0)=c_0(be_{-1})(0)=\langle f, be_{-1} (be_{-1})(0)\rangle.
$$
\end{remark}

Now we can iterate \eqref{hs22}:
\begin{equation*}\label{hs00023}
b^{-2}H^2(\alpha\mu^2)=\{e_{-2}\}\oplus 
\{e_{-1}\}\oplus H^2(\alpha),
\end{equation*}
and so on...

\begin{theorem}
The system of vectors 
\begin{equation}\label{hs23}
e_n(\z)=e_n^\alpha(\z):=b^n \frac{k^{\alpha\mu^{-n}}(\z)}{\sqrt{k^{\alpha\mu^{-n}}(0)}}, 
\   n\in \bbZ,
\end{equation}
forms an orthonormal basis in $L^2(\alpha)$.
\end{theorem}

\begin{proof}
Due to Lemma \ref{l42} we have to show that this system approximate any vector $f$ form $L^2(\alpha)\ominus H^2(\alpha)$.  Since $\check H^2(\alpha)\subseteq H^2(\alpha)$ it is of the form
$f=\Delta \bar g$, where $g\in \hat H_0^2(\alpha^{-1}\nu)$. By Lemma \ref{l42}, this vector can be successfully approximated by a finite combination
$$
g_n=\sum_{j=1}^n  a_j \hat e^{\alpha^{-1}\nu}_j.
$$
It means that for an arbitrary $\epsilon>0$ there exists $f_n=\Delta \bar g_n$ such that $\|f-f_n\| \le \epsilon$.

Due to duality \eqref{hs14}
$$
f_n=\sum_{j=1}^n\bar a_j\Delta\overline{\hat e^{\alpha^{-1}\nu}_j}=
\sum_{j=-n}^{-1}\bar a_j  \check e^{\alpha}_j\in b^{-n} \check H^2(\alpha\mu^n).
$$
Since $\check H^2(\alpha\mu^n)\subseteq H^2(\alpha\mu^n)$, this vector is approximated by
span$_{j\ge -n}\{e_j\}$.
\end{proof}

Our construction, see \eqref{00hs21}, leads to the following theorem.
\begin{theorem}\label{th44}
The multiplication operator by $\fz$ in $L^2(\alpha)$ with respect to the basis
\eqref{hs23} is the following Jacobi matrix $J=J(H^2(\alpha))$:
\begin{equation}\label{hs24}
\fz e^{\alpha}_n=p_n(\alpha) e^{\alpha}_{n-1}+
q_n(\alpha) e^{\alpha}_{n}+
p_{n+1}(\alpha) e^{\alpha}_{n+1},
\end{equation}
where
\begin{equation}\label{hs25}
p_n(\alpha)=\cP(\alpha\mu^{-n}), \quad \cP(\alpha)=(\fz b)(0)\sqrt{\frac{k^{\alpha}(0)}{k^{\alpha\mu}(0)}}
\end{equation}
and
\begin{equation}\label{00hs25}
q_n(\alpha)=\cQ(\alpha\mu^{-n}), \ 
\cQ(\alpha)=\frac{(\fz b)(0)}{b'(0)}\left\{\frac{(k^{\alpha})'(0)}{k^{\alpha}(0)}
-\frac{(k^{\alpha\mu})'(0)}{k^{\alpha\mu}(0)}
\right\}+\frac{(\fz b)'(0)}{b'(0)}.
\end{equation}

\end{theorem}
\begin{proof}
The exact formulas for coefficients are obtained  by 
expansion of the involved functions in the origin.
\end{proof}

\subsection{The dual basis and reflectionless property}
By Theorem \ref{th44} the inclusion \eqref{hs16} implies
$$
\check H^2_0(\alpha^{-1}\nu)\subseteq \Delta \overline{L^2(\alpha)\ominus H^2(\alpha)}
\subseteq \hat H^2_0(\alpha^{-1}\nu).
$$
So, if we define the dual Hardy  space
$$
b\tilde H^2(\mu^{-1}\alpha^{-1}\nu):=\Delta \overline{L^2(\alpha)\ominus H^2(\alpha)}
$$
then in the dual basis
\begin{equation}
\label{hsd0}
b(\z)\tilde e_n(\z)=\Delta(\z) \overline {e_{-n-1}^\alpha(\z)},\quad \z\in\bbT,
\end{equation}
we obtain the Jacobi matrix $\tau J=J(\tilde H^2(\mu^{-1}\alpha^{-1}\nu))$, which entries are related to the initial matrix $J=J(H^2(\alpha))$ by
$$
\tau p_n=p_{-n},\quad \tau q_n=q_{-n-1}.
$$

\begin{theorem}
The Jacobi matrix $J=J(H^2(\alpha))$ is reflectionless. Moreover its resolvent functions $r_\pm(z)$, see \eqref{kli3}, are of the form
\begin{equation}\label{hsd1}
r_+(\fz(\z))=-\frac{e_0(\z)}{p_0 e_{-1}(\z)}, \quad r_-(\fz(\z))=-\frac{\tilde e_0(\z)}{p_0 \tilde e_{-1}(\z)}.
\end{equation}

\end{theorem}
\begin{proof}
First of all $J$ is generated by the multiplication operator \eqref{hs24} and therefore its spectrum is $E$.

It is well known that
$$
r_+(z)=-\cfrac{1}{z-q_0-\cfrac{p_1^2}{z-q_1-\dots}}
$$
The same expansion we obtain from the recurrence relation \eqref{hs24} for the ratio
$-\frac{e_0(\z)}{p_0 e_{-1}(\z)}$. We apply this argument to the dual matrix $\tau J$ to get the second formula in \eqref{hsd1}. Now we use \eqref{hsd0} to prove \eqref{kli4}.
\end{proof}

\section{Parametrization of $J(E)$ by collections of \\ spectral data: divisors $D(E)$}

In this section we do not assume that $\Omega:=\Om$ is of Widom type, that is, here $E$ is an arbitrary closed set without isolating points, and $|E|>0$.

We start with the formula, which connects $r_\pm$ with the resolvent of the whole matrix $J$. The vectors $e_{-1}$ and $e_0$ form the cyclic subspace for $J$. The so called matrix resolvent function is defined by
\begin{equation}\label{kli5}
\begin{bmatrix}
R_{-1,-1}& R_{-1,0}\\
R_{0,-1}&R_{0,0}
\end{bmatrix}(z)=\cE^*(J-z)^{-1}\cE,\quad \cE\begin{bmatrix} c_{-1}\\ c_0
\end{bmatrix}= c_{-1}e_{-1}+c_{0}e_{0},
\end{equation}
for $\cE:\bbC^2\to \l^2$. Then
\begin{equation}\label{kli6}
\begin{bmatrix}
R_{-1,-1}& R_{-1,0}\\
R_{0,-1}&R_{0,0}
\end{bmatrix}(z)=\begin{bmatrix} r^{-1}_{-}(z)& p_0\\ p_0&r_+^{-1}(z)
\end{bmatrix}^{-1}.
\end{equation}
In particular,
\begin{equation}\label{kli7}
-\frac{1}{R_{0,0}(z)}=-\frac 1 {r_+(z)}+p_0^2 r_-(z).
\end{equation}
Notice that $R_{0,0}(z)$ is of Nevanlinna class, moreover, it is analytic in $\bar \bbC\setminus E$ and real in
$\bbR\setminus E$. By \eqref{kli4} it assumes pure imaginary values (a.e.) in $E$. Due to well known formula we can restore $R_{0,0}(z)$ by its argument on the real axis
\begin{equation}\label{kli8}
R_{0,0}(z)=\frac{-1} {\sqrt{(z-a_0)(z-b_0)}}\prod_{j\ge 1}\frac{z-x_j} {\sqrt{(z-a_j)(z-b_j)}},
\end{equation}
where $x_j\in[a_j,b_j]$ is the unique point in which the argument of $R_{0,0}(x)$ switches from 
$\pi$ to $0$.

An element of $D(E)$ is a sequence $D=\{(x_j,\epsilon_j)\}_{j\ge 1}$, where $x_j\in[a_j,b_j]$ and $\epsilon_j=\pm 1$. By definition we identify the points $(a_j,1)=(a_j,-1)$ and 
$(b_j,1)=(b_j,-1)$.

\begin{definition}\label{def31} The map from $J(E)$ to $D(E)$ is defined as follows:  $x_j$'s are defined by \eqref{kli8}. We have to define $\epsilon_j$ only for $x_j\in (a_j,b_j)$: $\epsilon_j=1$ if $x_j$ is a pole of $1/r_+(z)$ and $\epsilon_j=-1$ if $x_j$ is a pole of $p_0^2 r_-$.
\end{definition}

 Notice that by \eqref{kli7} at least one of these functions must have a pole at $x_j$. On the other hand, still by \eqref{kli6},
$$
R_{-1,-1}(z)=r_+^{-1}(z) r_-(z) R_{0,0}(z)
$$
and analyticity of $R_{-1,-1}$ at $x_j$ implies that only one of them may have a pole at this point.

\begin{lemma}
Assume that $J(E)$ consists of matrices with absolutely continuous spectrum, then the map $J(E)\to D(E)$ is one to one.
\end{lemma}
\begin{proof}
This is an easy consequence of \cite[Theorem 1.3]{PR09}.
In particular, it says, that the measure related to the Nevanlinna function
$-1/R_{0,0}$ is absolutely continuous on $E$. Therefore both Nevanlinna functions
$-1/r_{+}$ and $p_0^2r_-$ have absolutely continuous measure on $E$, see \eqref{kli7}. By
\eqref{kli4} they are equal. Thus both functions are uniquely determine by the function
$-1/R_{0,0}$ and the collection $\{\e_j\}_{j\ge 1}$.
\end{proof}

We define on  $D(E)$ the topology of the direct product of $I_j=\{(x_j,\epsilon_j): x_j\in [a_j,b_j], \epsilon_j=\pm1\}$, and assume that each $I_j$ is homeomorphic to the circle $\bbT$.
Hence, $J(E)$ and $D(E)$ are homeomorphic.

\section{Generalized Abel map}

In this section we will use actively certain analytic fact: 
\begin{theorem}[Sodin, Yuditskii \cite{SY}]
\label{thSY}
 Let $\Omega$ be a Widom domain. Let $F$ be a Nevanlinna class  function in the upper half-plane and meromorphic in $\Omega$, and let its poles satisfy the Blaschke condition in $\Omega$.
 Then $F\circ \fz$ is of bounded characteristic in $\bbD$, whose inner component is only ratio of two Blaschke products.
\end{theorem}

\subsection{Formula for reproducing kernel $k^{\alpha}$ and the divisor. \\ Wronski's formula}
\label{wronski}
Based on the fact that $J=J(H^2(\alpha))$ is reflectionless we prove two lemmas describing properties of reproducing kernels. But first we formulate a certain corollary of Theorem \ref{thSY}.
\begin{corollary}
For a collection $\{x_j\}_{j\ge 1}$, $x_j\in [a_j,b_j]$, let 
\begin{equation}\label{abm1}
W(z)=\prod_{j\ge 1}\frac{z-x_j}{z-c_j},
\end{equation}
where $c_j$ are the critical points. Then the inner part of $W\circ\fz$ is the ratio $\prod_{j\ge}b_{x_j}/\Delta$, and $\Delta (W\circ\fz)/\prod_{j\ge}b_{x_j}$ is its outer part.
\end{corollary}
\begin{proof}
We can represent $W$ as a ratio of two Nevanlinna functions, say 
$$
\prod_{j\ge 1}\frac{z-x_j}{z-a_j}\ \text{and }\prod_{j\ge 1}\frac{z-c_j}{z-a_j},
$$
and then apply Theorem \ref{thSY} to each of them.
\end{proof}

\begin{lemma}\label{le43}
Let $k^{\alpha}(\z)=e_0(\z)e_0(0)$ be the reproducing kernel. Then there exists a unique divisor $D=\{(x_j,\epsilon_j)\}\in D(E)$ such that, see \eqref{abm1},
\begin{equation}
\label{abm2}
e_0(\z)=\prod_{j\ge 1}b_{x_j}^{\frac{1+\epsilon_j} 2}
\sqrt{\frac{(W\circ\fz) \Delta(\z)}{\prod_{j\ge 1}b_{x_j}(\z)}}.
\end{equation}
In particular,
\begin{equation}\label{abm3}
{k^\alpha(0)}=e^2_0(0)=\Delta(0)\prod_{j\ge 1}b_{x_j}^{\epsilon_j}(0).
\end{equation}
\end{lemma}

\begin{lemma}\label{le44}
Let $e_{-1}$, $e_0$ be defined by \eqref{hs23} and $\tilde e_{-1}$, $\tilde e_0$ by \eqref{hsd0}.
Then
\begin{eqnarray}\label{abm4}
\lefteqn{p_0(e_{-1}\tilde e_{-1}(\z)-e_0(\z)\tilde e_0(\z))} \nonumber
\\&=&\frac{\Delta(\z)}{b(\z)}\sqrt{(\fz-a_0)(\fz-b_0)}
\prod_{j\ge 1}\frac{\sqrt{(\fz-a_j)(\fz-b_j)}}{\fz-c_j}.
\end{eqnarray}

\end{lemma}

\begin{proof}[Proof of Lemma \ref{le43}]
The Cauchy transform of the harmonic measure is of the form
\begin{equation*}
\int_E\frac{d\omega(x)}{x-z}=\frac{-1}{\sqrt{(z-a_0)(z-b_0)}}\prod_{j\ge 1}\frac{(z-c_j)}{\sqrt{(z-a_j)(z-b_j)}}.
\end{equation*}
It is absolutely continuous, and its density is
\begin{equation}\label{00abm4}
d\omega(x)=\frac 1\pi \prod_{j\ge 1}\frac{(x-c_j)}{\sqrt{(x-a_j)(x-b_j)}}\,
\frac{dx}{\sqrt{(x-a_0)(b_0-x)}}.
\end{equation}
Since $J=J(H^2(\alpha))$ is reflectionless we have by \eqref{kli8}
\begin{eqnarray}\label{abm5}
R_{0,0}(z)=\frac{-1}{\sqrt{(z-a_0)(z-b_0)}}\prod_{j\ge 1}\frac{z-x_j}{z-c_j}\frac{(z-c_j)}{\sqrt{(z-a_j)(z-b_j)}}.\nonumber
\\
=\int_E  \prod_{j\ge 1}\frac{x-x_j}{x-c_j}\frac{d\omega(x)}{x-z}.
\end{eqnarray}

On the other hand, due to the functional model
\begin{equation*}
R_{0,0}(z)=\langle (J-z)^{-1} e_0,e_0\rangle=\int_{\bbT}\frac{|e_0(\z)|^2}{\fz-z}dm(\z),
\end{equation*}
or, by \eqref{hs6},
\begin{equation}\label{abm6}
R_{0,0}(z)=\int_{E}|e_0(\fz^{-1}(x))|^2\frac{d\omega(x)}{x-z}.
\end{equation}

Compare \eqref{abm5} and \eqref{abm6} to get
\begin{equation}\label{abm7}
|e_0(\z)|^2=W(\fz(\z)), \  \z \in\bbT.
\end{equation}
Therefore the outer part of $e_0(\z)$ is
$\sqrt{\frac{(W\circ\fz) \Delta(\z)}{\prod_{j\ge 1}b_{x_j}(\z)}}$.

Now, by \eqref{hsd0} and \eqref{abm7}, we have
\begin{equation*}\label{00abm8}
e_0(\z) b(\z)\tilde e_{-1}(\z)=\Delta(\z)W(\fz(\z)), \  \z \in\bbT.
\end{equation*}
Since these functions of Smirnov's class coincide on $\bbT$ they are equal in $\bbD$. That is, the inner part of $e_0(\z)$ is a divisor of $\prod_{j\ge 1}b_{x_j}$. Therefore we obtain the formula \eqref{abm2} with a certain choice of $\e_j=\pm 1$.

Let us note that poles of $r_+^{-1}(\fz(\z))$ are zeros of $e_0(\z)$, see \eqref{hsd1}. That is, our choice of $\epsilon_j$'s is exactly  the same as in Definition \ref{def31}.
\end{proof}

\begin{proof}[Proof of Lemma \ref{le44}]
We use the second relation in \eqref{kli6}:
\begin{equation*}
R_{-1,0}(z)=\frac{-p_0}{r^{-1}_+(z)r^{-1}_-(z)-p_0^2},
\end{equation*}
or, by \eqref{hsd1}
\begin{equation}\label{abm8}
\displaystyle
R_{-1,0}\circ\fz=\begin{cases}\frac 1{p_0}\frac{e_0\tilde e_0}{e_0\tilde e_0-e_{-1}\tilde e_{-1}}(\z),&\z\in \bbD\\
\frac 1{p_0}\frac{e_0\bar e_{-1}}{e_0 \bar e_{-1}-e_{-1}\bar e_{0}}(\z), &\z\in\bbT
\end{cases}
\end{equation}

Again, based on the functional model we have
\begin{equation*}
R_{-1,0}(z)=R_{0,-1}(z)= \int_E e_0 \bar e_{-1}\frac{d\omega(x)}{x-z}
=\int_E\bar e_0  e_{-1}\frac{d\omega(x)}{x-z},
\end{equation*}
or
\begin{equation}\label{abm9}
R_{-1,0}(z)=\int_E \Re(e_0 \bar e_{-1})\frac{d\omega(x)}{x-z}
\end{equation}
\end{proof}
Since the density in \eqref{abm9} corresponds to the imaginary part of $R_{-1,0}$ in \eqref{abm8} on $\bbT$ we have
$$
\frac{1}{\pi p_0}\frac {1}{\Im \bar e_0  e_{-1}}=\omega'\circ\fz.
$$
Now we use \eqref{00abm4}. Taking reciprocal and multiplying by $\Delta$ we get \eqref{abm4}.

\subsection{Canonical factorization and the Abel map $\pi: J(E)\to \Gamma^*$}
\label{canonical}
The following analytic lemma plays the key role in a definition of the generalized Abel map for reflectionless matrices with a Widom resolvent domain.

\begin{lemma}
\label{lemma45}
Let $\Omega=\Om$ be a Widom domain and $J$ be a reflectionless matrix with the spectrum $E$.
Then there exists a unique factorization
\begin{equation}\label{abm10}
r_+\circ\fz=-\frac 1{p_0}\frac{e_0}{e_{-1}}
\end{equation}
such that
\begin{eqnarray}\label{abm11}
\lefteqn{p_0(e_{-1}(\z)\overline{ e_{0}(\z)}-e_0(\z)\overline{ e_{-1}(\z)})} \nonumber
\\&=&\sqrt{(\fz-a_0)(\fz-b_0)}
\prod_{j\ge 1}\frac{\sqrt{(\fz-a_j)(\fz-b_j)}}{\fz-c_j},\  \fz\in \bbT,
\end{eqnarray}
where $e_{0}$ and $b e_{-1}$ are of Smirnov class with mutually simple singular parts and $e_0(0)>0$.
Moreover $e_0(\z)$ is of the form \eqref{abm2} for the divisor $D$ given in Definition \ref{def31}.
\end{lemma}

Of course, the condition \eqref{abm11} here reflects the property \eqref{abm4}, which was established in Lemma \ref{le44} for the reproducing kernels in $H^2(\alpha)$.

\begin{proof}
First assume that $r_+$ is of the form \eqref{abm10} and \eqref{abm11} holds. Since
\begin{equation}\label{dlyasasha}
-\Im \frac 1{R_{0,0}}=-2\Im\frac 1{r_+}=\frac{p_0(e_{-1}(\z)\overline{ e_{0}(\z)}-e_0(\z)\overline{ e_{-1}(\z)})}{i|e_0(\z)|^2}
\end{equation}
we have  $|e_0(\z)|^2=W\circ\fz$, where $W$ is of the form \eqref{abm1}. This defines uniquely the outer part of $e_0(\z)$. Further, the function $r_+$ is of the Nevanlinna class function in $\mathbb{C}_+$ and all its poles lie in the lacunas, not more than one in a gap $(a_j,b_j)$. Therefore we can apply Theorem \ref{thSY}, according to which the inner part of $r_+$ is the ratio of two Blaschke products. Thus,
by \eqref{abm10}, the numerator of this ratio  is the inner part of $e_0(\z)$. As the result we obtain $e_0(\z)$ in the form 
\eqref{abm2}. Then $p_0e_{-1}$ is defined uniquely by $p_0e_{-1}(\z)=-e_0(\z)/r_+\circ\fz$.

 Conversely, we define $e_0(\z)$ by \eqref{abm2} and then $p_0e_{-1}(\z)$ by \eqref{abm10}. It remains to check 
 \eqref{abm11}. But it follows from \eqref{dlyasasha} and $|e_0(\z)|^2=W\circ\fz$, which holds for the given $e_0(\z)$.
\end{proof}

\begin{definition}
The Abel map $\pi: J(E)\to \Gamma^*$ is defined for Widom domain as the character $\alpha\in\Gamma^*$ of the function
$e_0(\z)$ in the representation \eqref{abm2}.
\end{definition}

\begin{remark} Note that by \eqref{kli6}
\begin{equation*}\label{dlyasashb}
-\Im \frac 1{R_{-1,-1}}=- {|p_0 r_+|^2}\Im \frac 1{R_{0,0}}=\frac{p_0(e_{-1}(\z)\overline{ e_{0}(\z)}-e_0(\z)\overline{ e_{-1}(\z)})}{i|e_{-1}(\z)|^2}.
\end{equation*}
Therefore $e_{-1}(\z)$ is also of the form \eqref{abm2}. Thus both functions allow to  define the dual functions 
$\tilde e_0 (\z)$ and $(b\tilde e_{-1})(\z)$ of  the Smirnov class   by
\begin{equation*}
\label{hsd0again}
b(\z)\tilde e_n(\z)=\Delta(\z) \overline {e_{-n-1}(\z)},\quad \z\in\bbT, \, n=-1, 0\,.
\end{equation*}
Then, by reflectionless, 
\begin{equation}
\label{rminus}
r_-\circ \fz = -\frac{\tilde e_0}{p_0 \tilde e_{-1}}\,.
\end{equation}
\end{remark}

\begin{remark}
The absolutely continuity of the spectrum  $J\in J(E)$ was not used so far.
\end{remark}

\subsection{Spectral Theory and the Hardy space related to $J\in J(E)$}

For an arbitrary two-sided Jacobi matrix $J$ span$\{e_{-1},e_0\}$ is a cyclic subspace. The spectral  $2\times 2$ matrix  measure $d\sigma$ is defined by
$$
R(z)=\cE^*(J-z)^{-1}\cE=\int\frac{d\sigma(x)}{x-z},
$$
see \eqref{kli5}. Then the standard basis vectors of $\l^2$ corresponds to
\begin{equation}\label{abm12}
e_n\to \begin{bmatrix}
-p_0 Q^+_n\\P_n^+
\end{bmatrix}, \quad 
e_{-n-1}\to \begin{bmatrix}
P_n^{-}\\-p_0 Q^-_n
\end{bmatrix},
\end{equation}
where $P_n^\pm$ and $Q^\pm_n$ are orthogonal polynomials of the first and second kind generated by $J_{\pm}$.
The operator $J$ becomes the operator multiplication by independent variable in $L^2_{d\sigma}$.

Let $J$ be reflectionless with the Widom resolvent domain. We use the general formula \eqref{kli6} and plug in the representation
\eqref{abm10} and \eqref{rminus} to get
\begin{equation}\label{abm13}
R\circ\fz=\Psi\Phi^{-1},
\end{equation}
where
\begin{equation}\label{abm14}
\Psi=-p_0\begin{bmatrix}\tilde e_0&0\\
0& e_0
\end{bmatrix},\quad 
\Phi=\begin{bmatrix} \tilde e_{-1}&-e_0\\
-\tilde e_0& e_{-1}
\end{bmatrix}.
\end{equation}

\begin{lemma}
Let $J\in J(E)$ and $\Om$ be of Widom type. Then the matrix spectral density $\sigma'_{a.c.}$ is of the form
\begin{equation}\label{abm15}
\sigma_{a.c.}'=\frac{1}{2 \omega'}(\Phi^{-1})^*\Phi^{-1}.
\end{equation}
\end{lemma}

\begin{proof}
We use $\sigma_{a.c.}'=\frac 1{2\pi i}(R-R^*)$. We plug in \eqref{abm13} 
$$
\sigma_{a.c.}'=\frac 1{2\pi i}(\Phi^{-1})^*(\Phi^*\Psi-\Psi^*\Phi)\Phi^{-1},
$$
and use \eqref{abm11}.
\end{proof}

\begin{theorem}\label{thspectral10}
Assume, in addition,  that $J\in J(E)$ has absolutely continuous spectrum. Then
the map 
\begin{equation}\label{amb16}
\begin{bmatrix} P_{-1}(x)\\ P_0(x)
\end{bmatrix}\to f(\z)=e_{-1}(\z) P_{-1}\circ\fz+e_0(\z) P_0\circ\fz, \quad \begin{bmatrix} P_{-1}(x)\\ P_0(x)
\end{bmatrix}\in L^2_{d\sigma},
\end{equation}
acts unitary from $L^2_{d\sigma}$ to $L^2(\alpha)$, $\alpha=\pi(J)$. Moreover, the composition map 
\begin{equation}\label{amb17}
\cF: \l^2\to L^2_{d\sigma}\to L^2(\alpha)
\end{equation}
is such that $H^2_J:=\cF(\l^2_+)$ possesses the properties
\begin{equation}\label{amb18}
\check H^2(\alpha)\subseteq H^2_J\subseteq \hat H^2(\alpha),
\quad \fz (H^2_J)_0\subset H^2_J.
\end{equation}
In other words, this $J$ is one of those built in Subsection \ref{subsec25}.
\end{theorem}

\begin{proof}
First of all we note that if $f(\z)$ is of the form \eqref{amb16} then
\begin{equation}\label{amb19}
\begin{bmatrix}
f(\z)\\
\Delta(\z) f(\bar \z)/b(\z)
\end{bmatrix}=\begin{bmatrix}
e_{-1}&e_0\\
\tilde e_0&\tilde e_{-1}
\end{bmatrix}\begin{bmatrix} P_{-1}\\ P_0
\end{bmatrix}\circ\fz= \Phi^{-1}
\begin{bmatrix} P_{-1}\\ P_0
\end{bmatrix}\circ\fz\cdot \det \Phi.
\end{equation}
Therefore, by \eqref{abm15} and \eqref{abm4}, we have
$$
\frac 1 2 (\int_\bbT |f(\z)|^2 dm(\z)+\int_\bbT |f(\bar \z)|^2 dm(\z))=
\int_E\begin{bmatrix} P_{-1}\\ P_0
\end{bmatrix}^*(x) d\sigma(x) \begin{bmatrix} P_{-1}\\ P_0
\end{bmatrix}(x).
$$
That is, this map is an isometry.

Let 
$$
e_n(\z)=\cF (e_n)=-p_0 e_{-1}(\z) Q_n^+(\fz)+e_0(\z)P_n^+(\fz),\  n\ge 0.
$$
Since $-Q_n^+/P_n^+$ is the Pade approximation for $r_+$, this function has zero of exact  multiplicity $n$ at the origin. Thus, it is of Smirnov class, and therefore belongs to $\hat H^2(\alpha)$. It proves
$H^2_J\subseteq H^2(\alpha)$ and $\fz (H^2_J)_0\subset H^2_J$.

To show that $\check H^2(\alpha)\subset H^2_J$ we pass to the dual representation for the orthogonal complement
$$
\Delta(\z) e_{-n-1}(\bar \z)/b(\z)=-p_0 \tilde e_{-1}(\z) Q_n^-(\fz)+\tilde e_0(\z)P_n^-(\fz),\  n\ge 0.
$$

\end{proof}

\section{$j$-inner matrix functions and hidden singularity }

In what follows we always assume that all three assumptions $(i-iii)$ of Theorem \ref{thmain} on the set $E$ are fulfilled.

Recall that the orthogonal polynomials of the first $P_k(z)$ and second $Q_k(z)$ kind are given by the recurrence relations:
\begin{equation}\label{transfer1}
\begin{bmatrix}
p_0&      & & & \\
q_0-z&p_1& & & \\
p_1& q_1-z& p_2 & & \\
    & \ddots&\ddots&\ddots&\\
    &          &    p_n& q_n-z&p_{n+1}
\end{bmatrix}
\begin{bmatrix}
P_0&Q_0\\
P_1&Q_1\\
P_2&Q_2\\
\vdots&\vdots\\
P_{n+1}& Q_{n+1}
\end{bmatrix}
=
\begin{bmatrix}
p_0 &0\\
0&1\\
0&0\\
\vdots&\vdots\\
0& 0
\end{bmatrix}
\end{equation}
Therefore the so called transfer matrix
\begin{equation*}
\fA(z):=\begin{bmatrix}
P_n(z)&Q_n(z)\\
p_{n+1}P_{n+1}(z)&p_{n+1} Q_{n+1}(z)
\end{bmatrix}
\end{equation*}
has a kind of resolvent representation
\begin{equation}\label{transfer2}
\fA(z)=
\begin{bmatrix}
0 &\hdots&0&1&0\\
0 &\hdots&0&0&p_{n+1}
\end{bmatrix}
(T-zS)^{-1}
\begin{bmatrix}
p_0 &0\\
0&1\\
0&0\\
\vdots&\vdots\\
0& 0
\end{bmatrix},
\end{equation}
where $(T-zS)$ is the matrix of the system \eqref{transfer1}.

The Christoffel-Darboux  identities
\begin{equation}\label{transfer3}
\frac{\fA(z)^*j\fA(z)-j}{z-\bar z}
=\sum_{k=1}^{n}\begin{bmatrix} \overline{P_k(z)}\\
 \overline{Q_k(z)}
 \end{bmatrix}
 \begin{bmatrix}{P_k(z)}&
 {Q_k(z)}
 \end{bmatrix}, \quad
j=\begin{bmatrix}
0&1\\
-1&0
\end{bmatrix}.
\end{equation}
shows that $\fA(z)$ is $j$-expanding in the upper half-plane, i.e., the given in \eqref{transfer3} expression  is a positive matrix, and it is $j$-unitary on the real axis, i.e.,  $\fA(z)^*j\fA(z)-j=0$, $z\in \bbR$.

Let us  define the transfer matrix between $\hat H^2(\alpha)$ and $\check H^2(\alpha)$ spaces by the relation 
\begin{equation}\label{transfer4}
\begin{bmatrix} \hat p_0\hat e_{-1}(\z)& \hat e_{0}(\z)
\end{bmatrix}=\begin{bmatrix} \check p_0 \check e_{-1}(\z)& \check e_{0}(\z)
\end{bmatrix}
\fA(\fz(\z)).
\end{equation}
Note that by the duality \eqref{hsd0}
$$
\begin{bmatrix}\hat p_0 \hat e_{-1}(\z)&\hat e_{0}(\z)\\
\hat p_0 \widetilde{\hat e}_{0}(\z)& \widetilde{\hat e}_{-1}(\z)
\end{bmatrix}
=\begin{bmatrix} \check p_0\check e_{-1}(\z)& \check e_{0}(\z)\\
\check p_0 \widetilde{\check e}_{0}(\z)& \widetilde{\check e}_{-1}(\z)
\end{bmatrix}
\fA\circ\fz
$$
and therefore due to the Wronski's formula \eqref{abm4} it is analytic in $\Omega$.  
\begin{remark}
Note that $\fA(z)$ meets at the infinity the following normalization
\begin{equation}\label{00transfer1}
\fA (z)=\begin{bmatrix}
1/\lambda+\dots &(1/\lambda -\lambda)\frac 1 z+\dots \\
0+\dots &\lambda+\dots
\end{bmatrix},
\end{equation}
where $\lambda=\langle \hat e_0,\check e_0\rangle=\frac{\check e_0(0)}{\hat e_0(0)}$. Indeed, due to \eqref{transfer4}
$$
(\frac{\hat e_0}{\check e_0}\circ\fz^{-1})(z)\begin{bmatrix} z+\dots & 1\end{bmatrix}=
\begin{bmatrix} z+\dots & 1\end{bmatrix}\fA(z).
$$
Since $\fA$ is regular, moreover the leading term is diagonal,  and $\det \fA(z)=1$ from this asymptotics  we get \eqref{00transfer1}. In addition,  the first non-vanishing term in the low-left corner is of the form
\begin{equation}\label{000transfer1}
\begin{bmatrix}
0&1
\end{bmatrix}
\fA\begin{bmatrix}
1\\0
\end{bmatrix}=-(\check p_0^2/\lambda-\lambda \hat p_0^2)\frac 1 z+\dots,
\end{equation}
which completes the duality generated by the tilde-involution. 
\end{remark}

Further, by 
\eqref{abm11} $\fA(z)$ is $j$-unitary 
\begin{equation}\label{transfer5}
j= \fA j\fA^*\quad \text{a.e. on} \ E.
\end{equation}
One can prove that it is $j$-inner, that is, in addition to \eqref{transfer5} one has
$$
\frac{\fA(z)j\fA^*(z)-j}{z-\bar z}\ge 0, \ z\in \Omega.
$$
To this end  we prove analogs of the representation \eqref{transfer2} and of the  Christoffel-Darboux  identity \eqref{transfer3}, see Appendix.

The main result of this section is the following theorem. It reveals  the structure of all possible subspaces $H^2(\alpha)$. They are related to $j$-inner divisors of $j$-inner matrix function $\fA$.

\begin{theorem}\label{mainj}
Let $\check H^2(\alpha)\subset H^2(\alpha)\subset \hat H^2(\alpha)$ and $\fz H_0^2(\alpha)\subset H^2(\alpha)$. Similarly to \eqref{transfer4} we define 
\begin{eqnarray}
\begin{bmatrix} \hat p_0\hat e_{-1}(\z)& \hat e_{0}(\z)
\end{bmatrix}&=&\begin{bmatrix}  p_0  e_{-1}(\z)&  e_{0}(\z)
\end{bmatrix}
\fA_1(\fz(\z)),\label{transfer14}
\\
\begin{bmatrix}  p_0 e_{-1}(\z)& e_{0}(\z)
\end{bmatrix}&=&\begin{bmatrix} \check p_0 \check e_{-1}(\z)& \check e_{0}(\z)
\end{bmatrix}
\fA_2(\fz(\z)).\label{transfer15}
\end{eqnarray}
Then $\fA(z)=\fA_2(z)\fA_1(z)$ and both matrices are $j$-inner. Conversely, let $\fA_1(z)$ be a 
$j$-inner divisor of $\fA(z)$, which meets the normalization
\begin{equation}\label{transfer16}
\fA_1(\infty)=\begin{bmatrix}
1/\lambda_1&0\\0&\lambda_1
\end{bmatrix}, \quad 
\left(\begin{bmatrix}1 &0\end{bmatrix} (z\fA_1)\begin{bmatrix}0\\ 1\end{bmatrix}
\right)(\infty)=\frac 1 {\lambda_1} -\lambda_1
\end{equation}
Then
\begin{equation}\label{transfer17}
r_+\circ\fz=-\frac{e_0(\z)}{p_0e_{-1}(\z)}, \ r_-\circ\fz=-\frac{\tilde e_0(\z)}{p_0\tilde e_{-1}(\z)},\ 
\Delta \bar e_{-1-k}=b\tilde e_k,
\end{equation}
are the resolvent functions related to a reflectionless matrix $J=J(H^2(\alpha))$. Moreover,
\begin{equation}\label{transfer18}
k^{\alpha}(\z,\z_0)=
\frac{\begin{bmatrix} p_0 e_{-1}(\z)&e_0(\z)\end{bmatrix}
\begin{bmatrix}0&1\\-1&0\end{bmatrix}
\begin{bmatrix} p_0\overline{e_{-1}(\z_0)}\\
\overline{e_0(\z_0)}\end{bmatrix}
}
{\fz(\z)-\bar z_0}
\end{equation}
is the reproducing kernel of the corresponding $H^2(\alpha)$.
\end{theorem}

In the proof we will use a well known general property of  Hilbert spaces of analytic functions with reproducing kernels, which we give here for a reader convenience.

\begin{lemma}\label{lark}
Let $H_O$ be a Hilbert space of analytic in a domain $O$ functions and $k_{\z_0}=k(\z,\z_0)$ be the reproducing kernel. Assume that an analytic  function $f$in $O_1\subset O$ is such that
\begin{equation}\label{transfer19}
|\sum f(\z_n)\bar \xi_n|^2\le C\sum_{n,m}k(\z_n,\z_m)\bar \xi_n\xi_m,
\end{equation}
 for all finite collections $\{\z_n\}$ in $O_1$. Then this function has an analytic extension in $O$, 
 $f(\z)=g(\z)$, $\z\in O_1$. Moreover, $g\in H_O$, $\|g\|^2\le C$.
\end{lemma}
\begin{proof}
Consider the densely defined functional
$$
\Lambda(\sum_m k_{\z_m}\xi_m)=\sum_m\overline{f(\z_m)}\xi_m
$$
By \eqref{transfer19} it is bounded. Therefore there exits $g\in H_O$, $\|g\|^2\le C$, such that
$\overline{f(\z)}=\Lambda(k_\z)=\langle  k_{\z}, g\rangle$, $\z\in O_1$. But $k_\z$ is the reproducing kernel. Thus 
$f(\z)=g(\z)$.
\end{proof}

\begin{proof}[Proof of Theorem \ref{mainj}]
For both pairs  $H^2(\alpha), \hat H^2(\alpha)$ and  $\check H^2(\alpha),  H^2(\alpha)$ we can define the transfer matrices,  find their representations in the form
\eqref{transfer8} and, therefore, prove that they are  $j$-inner by  \eqref{0transfer12}.

Conversely, let $\fA=\fA_2\fA_1$. 
We define $e_0, p_0 e_{-1}$ and then $k_{\z}^{\alpha}$ by \eqref{transfer18}.
Let us  note that by \eqref{transfer15} 
$$
k^\alpha(\z,\z)\ge \check k^\alpha(\z,\z)\ge 0.
$$
 That is,
$\Im r_+(z)\ge 0$ in the upper half-plane, and the normalization \eqref{transfer16} implies
$(zr_+)(\infty)=-1$, see \eqref{transfer17}. 

Our first goal is to show that $k^\alpha_{\z_0}$ is at least in $\hat H^2(\alpha)$. Note that $k^\alpha(\z,\z_0)$ is a positive definite kernel. In particular,
\begin{equation}\label{transfer20}
|\sum k_{\z_0}(\z_n)\bar \xi_n|^2\le k(\zeta_0,\z_0)\sum_{n,m}k(\z_n,\z_m)\bar \xi_n\xi_m,
\end{equation}
By \eqref{transfer18} the kernel $\hat k(\z,\z_0)$ dominates the kernel  $k(\z,\z_0)$,
\begin{equation}\label{transfer21}
\sum_{n,m}k(\z_n,\z_m)\bar \xi_n\xi_m \le \sum_{n,m}\hat k(\z_n,\z_m)\bar \xi_n\xi_m.
\end{equation}
 Thus
\eqref{transfer20}, \eqref{transfer21} and Lemma \ref{lark} imply that $k_{\z_0}\in \hat H^2(\alpha)$.
In particular,
$$
e_0 e_0(0)=k^\alpha_0\in \hat H^2(\alpha)\ \text{and}\ e_{-1}\overline{e_0(\zeta_0)}=\frac{\fz-\bar z_0}{p_0}k^{\alpha}_{\z_0}+e_0 \overline{e_{-1}(\z_0)}\in\frac 1 b \hat H^2(\alpha\mu).
$$
Note, by the way,  that  with respect to the norm we have just an estimation $\|k_{\z_0}\|_{\hat H^2(\alpha)}^2\le k(\z_0,\z_0)$.

In the same way we obtain similar properties for the dual functions
$\tilde e_0\in\hat H^2(\alpha^{-1}\mu^{-1}\nu)$ and $b\tilde e_{-1}\in \hat H^2(\alpha^{-1}\nu)$;
the function $r_-=-\tilde e_0/{p_0\tilde e_{-1} }$ is of Nevanlinna class.

Let us point out that \eqref{transfer14} defines the product $p_0e_{-1}(\z)$, but not the constant $p_0 $ and the function $e_{-1}(\z)$ along. The value of $p_0$ is fixed by the condition $(zr_-)(\infty)=-1$. 
Now, to the given Nevanlinna class functions $r_\pm$ we associate the one-sided Jacobi matrices
$J_\pm$. Together with the constant $p_0$ they form the two sided matrix $J$. 

For this $J$, the related resolvent function $R_{-1,-1}$ and $R_{0,0}$ are given by \eqref{kli6}. Since $\det\fA_1(z)=1$ the  functions $p_0e_{-1}, e_0, p_0\tilde e_0, \tilde e_{-1}$ satisfy  the Wronski's identity \eqref{abm4}. As it was shown these functions are analytic in $\Omega$, therefore the measures, corresponding the Nevanlinna class functions $R_{-1,-1}$ and $R_{0,0}$,  are supported on $E$. Thus the spectrum of $J$ is $E$, and $J\in J(E)$.

We still have $\| e_0\|_{L^2(\alpha)}\le 1$ and $\|e_{-1}\|_{L^2(\alpha)}\le 1$. However, the assumption that all Jacobi matrices of $J(E)$ have absolutely continuous spectrum implies that all of them are marked by $H^2(\alpha)$'s. That is, $J=J(H^2(\alpha))$.
Note that  $e_0$ and $be_{-1}$ are of Smirnov class, and due to the Wronski's identity  have mutually simple inner parts.
By Lemma \ref{lemma45} they correspond to the  reproducing kernels in $H^2(\alpha)$ and this finishes the proof. By the way we get
$\| e_0\|_{L^2(\alpha)}=\|e_{-1}\|_{L^2(\alpha)}= 1$.

\end{proof}

We finish this section with the following lemma.

\begin{lemma}
Let $\hat e^\alpha_0(0)=\check e^\alpha_0(0)$. Then $\hat H^2(\alpha)=\check H^2(\alpha)$.
\end{lemma}

\begin{proof}
Since $\check H^2(\alpha)\subseteq \hat H^2(\alpha)$, and $\hat k^{\alpha}(\z)=\hat e(\z)\hat e(0)$ is the reproducing kernel in it, we have
$$
\|\hat e_0-\check e_0\|^2=2-2\frac{\check e_0(0)}{\hat e_0(0)}=0.
$$
That is, $\hat e_0(\z)=\check e_0(\z)$. Then $\fA\begin{bmatrix}0\\1\end{bmatrix}=\begin{bmatrix}0\\1\end{bmatrix}$. Since $\det \fA=1$ this matrix is of the form
$$
\fA(z)=\begin{bmatrix}1&0\\B(z)& 1\end{bmatrix}.
$$

Since $\fA(z)$ is $j$-inner in the domain $\Omega$ we get that $B(z)$ is analytic in $\Omega$, belongs to the Nevanlinna class, $\Im B(z)/\Im z\ge 0$, and $\Im B(z)=0$ a.e. on $E$. Thus the corresponding   measure is singular and supported on $E$. 

On the other hand,
$$
-\hat r_+^{-1}(z)=-\check r_+^{-1}(z)+B(z).
$$
Since both functions $-\hat r_+^{-1}(z)$ and $-\check r^{-1}_+(z)$ correspond to the absolutely continuous measures on $E$, $B(z)=0$. Therefore $\hat J=\check J$ and $\hat H^2(\alpha)=\check H^2(\alpha)$.
\end{proof}

\begin{remark}
The same proof is valid in the case $H_1^2(\alpha)\subseteq H_2^2(\alpha)$ and $k^\alpha_1(0)=k^{\alpha}_2(0)$. Of course, it is important that the reflectionless measures are absolutely continuous. 
\end{remark}

\section{Shift invariant measure on $J(E)$. Ergodicity}

As it was shown in the previous section $\hat H^2(\alpha)=\check H^2(\alpha)$ if and only if
$\hat k^\alpha(0)=\check k^\alpha(0)$.

\begin{lemma}
$\hat k^{\alpha}(0)$ is an upper semi-continuous function on $\Gamma^*$.
\end{lemma}

\begin{proof}
Consider the system of normalized functions
$
\hat e^\alpha=\frac{\hat k^\alpha}{\sqrt{\hat k^\alpha(0)}}
$
and let $\alpha_n\to\beta$ be the extremal sequence, i.e.:
$$
\lim_{n\to \infty} \hat e^{\alpha_n}(0)=\limsup_{\alpha\to \beta} \hat e^\alpha(0).
$$
If needed we chose a subsequence, which converges in $H^2$ weakly,
$f=\lim_{n\to \infty} \hat e^{\alpha_n}$. The limit function belongs to $H^2(\beta)$ and its norm is less or equal than $1$. Therefore,
$$
\lim_{n\to \infty} \sqrt{\hat k^{\alpha_n}(0)}=f(0)=\langle f,\hat k^{\beta}\rangle\le \|f\|\|\hat k^{\beta}\|
\le \sqrt{\hat k^{\beta}(0)}.
$$
\end{proof}

Thus $\hat k^\alpha(0)$ and $\check k^{\alpha}(0)$ (it is lower semi-continuous, see \eqref{bishs14})
are measurable functions with respect to the Haar measure $d\alpha$ on $\Gamma^*$.
Let us point out that they are uniformly bounded, see \eqref{01hs14}. 

For a measurable set $X\subset\Gamma^*$ by $|X|_{d\alpha}$ we denote its Haar measure.

\begin{theorem}\label{th62}
Let
$$
\Theta=\{\alpha: \hat H^2(\alpha)\not=\check H^2(\alpha)\}.
$$
Then $|\Theta|_{d\alpha}=0$.
\end{theorem}

\begin{remark}
In the special case, when  the functional model for Jacobi matrices exists (see Section \ref{subsec25}) all main theorems for ergodic operators can be obtain from well known results in complex analysis, e.g. Thouless formula \eqref{thfo} can be obtained from the functional model. On the contrary  we were unable to derive Theorem \ref{th62} from general facts dealing with ergodic Jacobi matrices.
Our proof of this theorem is based on the following analytic  lemma. Note that by definition $\check H^2(\alpha)\subset
\hat H^2(\alpha)$. But it is far from being evident that there exists a multiplier $w$ such that
$w \hat H^2(\alpha)\subset \check H^2(\alpha)$, moreover that such a multiplier is uniform (does not depend on $\alpha$).
\end{remark}

\begin{lemma}[Main Lemma]
For an arbitrary $\beta\in\Gamma^*$ there exists a Blaschke product $w=\prod b_{x_j}$,
$x_j\in(a_j,b_j)$ such that
\begin{equation}\label{erg61}
w\hat H^2(\alpha)\subset \check H^2(\alpha\beta), \ \forall \alpha\in\Gamma^*.
\end{equation}
Also
\begin{equation}\label{erg62}
w(0)^2=\inf_{\alpha\in \Gamma^*}
\frac{\check k^{\alpha\beta}(0)}{\hat k^{\alpha}(0)}.
\end{equation}
\end{lemma}

\begin{proof} We define a functional
$\Lambda(f)=f(0)$ on $\hat H^1(\nu\beta^{-1})$. It is bounded. Since $\hat H^1(\nu\beta^{-1})\subset L^1$ there exists $w_1\in L^\infty$ such that
$$
\Lambda (f)=\langle f, w_1\rangle\ \text{and}\ \|\Lambda\|=\|w_1\|_\infty.
$$
Using the averaging operator \eqref{hs7}, we get the same representation
with $w_2=P^{\nu\beta^{-1}} w_1$ of the same norm. Since $\Lambda$ annihilates $\hat H^1_0(\nu\beta^{-1})$, by the Inverse Cauchy Theorem (Theorem \ref{thICT}), $w_2=\Delta \bar w_3$, where $w_3\in H^\infty(\beta)$. Thus,
$$
\Lambda(f)=\int_\bbT\frac{w_3}\Delta f dm=f(0),\, f \in \hat H^1(\nu\beta^{-1})\,.
$$
Let $f=\Delta g$, $g\in H^\infty(\beta^{-1})$, $g(0)\not=0$. Then
$$
(\Delta g)(0)=\int_\bbT w_3 g dm=(w_3 g)(0),
$$
i.e., $w_3(0)=\Delta(0)$.

We set $w:=\frac{w_3}{\|w_3\|_\infty}$. We can summarize its properties as follows: it is a $H^\infty(\beta)$-function of the norm one, such that
\begin{equation}\label{1part2}
\int_\bbT\frac{w}\Delta f dm=\frac{w(0)}{\Delta(0)}f(0), \quad \forall f\in \hat H^1(\nu\beta^{-1}),
\end{equation}
with the largest possible (for such functions) value $w(0)>0$. 

Let us prove \eqref{erg62}. Let us fix $\alpha\in \Gamma^*$, and 
\begin{equation}
\label{short}
h_1=h_1(0)\frac{\hat k^\alpha}{\hat k^\alpha(0)},\,
h_2=h_2(0)\frac{\hat k^{\alpha^{-1}\beta^{-1}\nu}}{\hat k^{\alpha^{-1}\beta^{-1}\nu}(0)}.
\end{equation}
Since $f=h_1h_2\in \hat H^1(\beta^{-1}\nu)$ we get by \eqref{1part2}
\begin{eqnarray*}
\frac{w(0)}{\Delta(0)}|h_1(0) h_2(0)|&\le \|f\|_1\le \|h_1\|_2\|h_2\|_2=
\sqrt{\left(
\frac{|h_1(0)|^2}{\hat k^\alpha(0)}
\right)
\left(
\frac{|h_2(0)|^2}{\hat k^{\alpha^{-1}\beta^{-1}\nu}(0)}
\right)}.
\end{eqnarray*}
Therefore, due to \eqref{hs14},
$$
w^2(0)\le\frac{1}{\hat k^{\alpha}(0)}\frac{\Delta(0)}{\hat k^{\alpha^{-1}\beta^{-1}\nu}(0)}=
\frac{\check k^{\alpha\beta}(0)}{\hat k^{\alpha}(0)}
$$
and 
$$
w^2(0)\le \inf_{\alpha\in \Gamma^*}\frac{\check k^{\alpha\beta}(0)}{\hat k^{\alpha}(0)}.
$$

But we need to prove equality.
Choose arbitrary $f\in \hat H^1(\nu\beta^{-1})$. Factorize it to $h_1h_2$ , where for  some $\alpha$, $ h_1\in \hat H^2(\alpha)$, $h_2\in \hat H^2(\alpha^{-1} \beta^{-1}\nu)$. We can write
\begin{eqnarray}
h_1&=&h_1(0)\frac{\hat k^\alpha}{\hat k^\alpha(0)}+b\tilde h_1, \quad \tilde h_1\in \Hat H^2(\alpha\mu^{-1}),\label{2part2}\\
h_2&=&h_2(0)\frac{\hat k^{\alpha^{-1}\beta^{-1}\nu}}{\hat k^{\alpha^{-1}\beta^{-1}\nu}(0)}+b\tilde h_2, \quad \tilde h_2\in \Hat H^2({\alpha^{-1}\beta^{-1}\nu}\mu^{-1}).\label{3part2}
\end{eqnarray}
Then using \eqref{2part2} and \eqref{3part2}, we see that
$$
\| f\|_1\ge \frac{|h_1(0) h_2(0)|}{\sqrt{\hat k^{\alpha}(0){\hat k^{\alpha^{-1}\beta^{-1}\nu}(0)}}}
=\frac{|f(0)|}{\Delta(0)}\sqrt{\frac{\check k^{\alpha\beta}(0)}{\hat k^{\alpha}(0)}}.
$$

 Notice that the norm of the functional in  \eqref{1part2} is $1$. Therefore,
$$
1\le \sup_{\alpha} \frac{w(0) |f(0)| \Delta(0)^{-1}}{|f(0)|\Delta(0)^{-1} \sqrt{\check k^{\alpha\beta}(0)} \sqrt{\hat k^{\alpha}(0)^{-1}}}\,. 
$$
Hence,
$$
w^2(0)\ge \inf_{\alpha\in \Gamma^*}\frac{\check k^{\alpha\beta}(0)}{\hat k^{\alpha}(0)}.
$$



Now we prove the first claim of the lemma. Since $\hat k^\alpha(0)$ is upper semi-continuous and 
$\check k^{\alpha\beta}(0)$ is lower semi-continuous their ratio is lower semi-continuous and the infimum is attained. Let $\alpha_0$ be the character such that
\begin{equation}\label{5part2}
w^2(0)=\frac{\check k^{\alpha_0\beta}(0)}{\hat k^{\alpha_0}(0)}.
\end{equation}

We claim that
\begin{equation}\label{6part2}
\langle(1-|w|^2) \hat e^{\alpha_0},\hat e^{\alpha_0}\rangle+
\|w \hat e^{\alpha_0}-\check e^{\alpha_0\beta} \|^2=0.
\end{equation}
Indeed, it can be simplified to
$$
\| \hat e^{\alpha_0}\|^2+\|\check e^{\alpha_0\beta}\|^2-
2\Re \langle w\hat e^{\alpha_0},\check e^{\alpha_0\beta}\rangle=
2 -2\Re \int_\bbT w\hat e^{\alpha_0}\overline{ \check e^{\alpha_0\beta}}dm.
$$
We use again \eqref{hs14}, and \eqref{5part2}, \eqref{1part2}   to obtain
$$
\int_\bbT w\hat e^{\alpha_0}\overline{ \check e^{\alpha_0\beta}}dm=
\int_{\bbT}\frac{w}{\Delta}\hat e^{\alpha_0}\hat e^{\alpha_0^{-1}\beta^{-1}\nu}dm=
\left(\frac{w}{\Delta}\hat e^{\alpha_0}\hat e^{\alpha_0^{-1}\beta^{-1}\nu}\right)(0)=1,
$$
since
$$
\left(\hat e^{\alpha_0}\hat e^{\alpha_0^{-1}\beta^{-1}\nu}\right)(0)=\frac{\hat e^{\alpha_0}(0)\Delta(0)}{\check e^{\alpha_0\beta}(0)}=
\Delta (0)\sqrt{\frac{\hat k^{\alpha_0}(0)}{\check k^{\alpha_0\beta}(0)}}.
$$

Thus, both nonnegative summands  in \eqref{6part2} are equal to 0. Since $\hat e^{\alpha}\not=0$ a.e. on $\bbT$ we have $|w|^2=1$ (a.e.). The second term implies
$
w=\frac{\check e^{\alpha_0\beta}}{\hat e^{\alpha_0}}.
$
Since $w$ is unimodular it is the ratio of inner parts of $\check e^{\alpha_0\beta}$
and $\hat e^{\alpha_0}$. We know that their inner parts are Blaschke products (see \eqref{abm2}). But $w$ is bounded analytic. Therefore there is at most one $x_j$ in every open gap $(a_j,b_j)$ such that $w=\prod b_{x_j}$. 

It remains to show that $w \hat H^2(\alpha)\subset \check H^2(\alpha\beta)$ for all $\alpha\in\Gamma^*$. In other words we have to show that $w \hat H^2(\alpha)$ is annihilated by functions from $\Delta\overline{\hat H_0^2(\alpha^{-1}\beta^{-1}\nu)}$. By the key property of $w$ \eqref{1part2} we have
$$
\langle w h_1, \Delta \bar h_2\rangle=\int_\bbT\frac{w}{\Delta}h_1 h_2 dm=0,
$$
where $h_1\in \hat H^2(\alpha)$ and $h_2\in \hat H_0^2(\alpha^{-1}\beta^{-1}\nu)$. Actually, this is the starting idea to consider the functional $\Lambda$ on $\hat H^1(\beta^{-1}\nu)$.
\end{proof}

\begin{remark}
The reader can notice that we used repeatedly the following equivalence. Function $w\in H^\infty(\beta)$ satisfies
$$
\int_{\mathbb{T}}\frac{w}{\Delta}f dm  = \frac{w(0)}{\Delta(0)} f(0)\,,\,\,\forall f\in \hat H^1(\beta^{-1}\nu)
$$
if and only if $w \hat H^2(\alpha) \subset \check H^2(\alpha\beta)$ for all $\alpha\in \Gamma^*$.
\end{remark}

\begin{lemma}
\label{incl}
Let $w, w(0)>0,$ be an inner character automorphic function with a character $\alpha$ and $w \hat H^2(\beta) \subset \check H^2(\alpha \beta) $. Then
$
w(0)^2  \le \frac{\check k^{\alpha\beta}(0)}{\hat k^\beta (0)}\,.
$
\end{lemma}

\begin{proof}
As $w \hat H^2(\beta)\subset \check H^2(\alpha\beta)$, we get 
$$
0\le \left\|w \frac {\hat k^{\beta}}{\|\hat k^{\beta} \|}-
\frac {\check k^{\alpha\beta}}{\|\check k^{\alpha\beta} \|}\right\|^2=
2-2w(0)\frac {\hat k^{\beta}(0)}{\|\hat k^{\beta}\|\|\check k^{\alpha\beta} \|}=2-2w(0)\frac {\|\hat k^{\beta}\|}{\|\check k^{\alpha\beta} \|}.
$$
Therefore,
$$
\frac {\hat k^{\beta}(0)}{\check k^{\alpha\beta} (0)}=
\frac {\|\hat  k^{\beta}\|^2}{\|\check k^{\alpha\beta} \|^2}\le\frac 1{w^2(0)}.
$$
\end{proof}

\begin{corollary}
Let $w=\prod_{j\ge 1} b_{x_j}$ be the function from $H^\infty(\mathbf{1})$, i.e., $w\circ\gamma=w$
for all $\gamma\in \Gamma$,  such that
$
w \hat H^2(\beta)\subset \check H^2(\beta).
$
We define $w_n=\prod_{j\ge n+1} b_{x_j}$ and denote the corresponding character by $\alpha_n$. 
Then 
\begin{equation}\label{kappa007}
\lim_{n\to\infty}\alpha_n=\mathbf{1},\quad \lim_{n\to\infty} w_n(0)=1,
\end{equation}
and the following uniform estimate holds
\begin{equation}\label{kappa1}
\frac {\hat k^{\beta}(0)}{\check k^{\alpha_n\beta} (0)}\le\frac 1{w^2_n(0)}, \quad \forall \beta\in\Gamma^*.
\end{equation}
\end{corollary}

\begin{proof} Since the product converges, $w_n$ converges to 1 on compact sets. It implies \eqref{kappa007}.
Let us prove \eqref{kappa1}.  Recall that by Corollary \ref{cor33} 
$\check H_0^2(\alpha)=b \check H^2(\alpha\mu^{-1})$. Similarly, $\{f\in\check H^2(\alpha): f(x)=0\}=b_x \check H^2(\alpha\mu_x^{-1})$ for all $x\in \Omega$. Functions from $b_{x_1}w_1 \hat H^2 (\beta)$ are all in $\check H^2(\beta)$ and they all are zeros at zeros of the Blaschke product $b_{x_1}$ whose character is  $\mu_{x_1}$. Then $w_1 \hat H^2 (\beta)$ are all in $\check H^2(\beta\mu_{x_1}^{-1})$. This means that $b_{x_2} w_2 \hat H^2 (\beta)\subset \check H^2(\beta\mu_{x_1}^{-1})$ and so on...
Finally we conclude that $w_n \check H^2(\beta)\subset \check H^2(\beta\mu_{x_1}^{-1}\dots \mu_{x_n}^{-1})= 
\check H^2(\beta\alpha_n)$.
Now we use Lemma \ref{incl} to conclude \eqref{kappa1}.
\end{proof}

\begin{proof}[Proof of Theorem \ref{th62}]
We have to show that $\check k^{\alpha}(0)=\hat k^\alpha(0)$ for almost all $\alpha$ with respect to the Haar measure. We define
$$
\k(\alpha)=\frac{\hat k^{\alpha}(0)}{\check k^{\alpha}(0)}
$$
and will show that $\k(\alpha)=1$ for almost all  $\alpha\in\Gamma^*$.
Since $\k(\alpha)\ge 1$ it is enough to prove that
\begin{equation}\label{7part2}
\int_{\Gamma^*}\k(\alpha)d\alpha=1.
\end{equation}

Using \eqref{kappa1} we get
\begin{equation}\label{8part2}
\int_{\Gamma^*}\k(\beta)d\beta=\int_{\Gamma^*}\frac{\hat k^{\beta}(0)}{\check k^{\alpha_n\beta} (0)} 
\frac{\hat k^{\alpha_n\beta} (0)}{\hat k^{\beta}(0)}d\beta\le \frac 1{w^2_n(0)}
\int_{\Gamma^*}
\frac{\hat k^{\alpha_n\beta} (0)}{\hat k^{\beta}(0)}d\beta.
\end{equation}
Let
$$
g_n(\beta)=\sup_{j\ge n}\frac{\hat k^{\alpha_j\beta} (0)}{\hat k^{\beta}(0)}
$$
Then $g_n(\beta)\ge g_{n+1}(\beta)$,
$$
\int_{\Gamma^*}
\frac{\hat k^{\alpha_n\beta} (0)}{\hat k^{\beta}(0)}d\beta\le\int_{\Gamma^*}g_n(\beta) d\beta
$$
and by B. Levi theorem there exists $g=\lim_{n\to \infty} g_n$ such that
$$
\lim_{n\to\infty}\int_{\Gamma^*} g_n(\beta) d\beta=\int_{\Gamma^*} g(\beta) d\beta.
$$

It is important that  $\lim_{n\to\infty}\alpha_n=\mathbf{1}$, see \eqref{kappa007}. 
It implies that $g\le 1$ by semi-continuity of $\hat k^\alpha(0)$. 
Indeed, 
$$
\lim_{n\to\infty} g_n(\beta)= \lim_{n\to\infty}\sup_{j\ge n}\frac{\hat k^{\alpha_j\beta} (0)}{\hat k^{\beta}(0)}\le \frac{\hat k^{\beta}(0)}{\hat k^{\beta}(0)}=1.
$$
Therefore
$$
\limsup_{n\to\infty}\int_{\Gamma^*}
\frac{\hat k^{\alpha_n\beta} (0)}{\hat k^{\beta}(0)}d\beta\le\lim_{n\to\infty}\int_{\Gamma^*}g_n(\beta) d\beta =\int_{\Gamma^*}g(\beta)d\beta\le 1.
$$
Now, $\lim_{n\to \infty }w_n(0)=1$   together with \eqref{8part2} imply
$\int_{\Gamma^*}\k(\beta)d\beta\le 1$. Evidently, $\int_{\Gamma^*}\k(\beta)d\beta\ge 1$.
Thus \eqref{7part2}, and with this, the theorem is proved.
\end{proof}

By {\it shift invariant measure} on $J(E)$ we understand the measure invariant under the transformation $J\rightarrow S^{-1} JS$, $J\in J(E))$.
As a consequence of Theorem \ref{th62} we prove the first part of the main Theorem \ref{thmain}.

\begin{theorem}
Let every $J\in J(E)$ have absolutely continuous spectrum, the domain $\Omega=\Om$ be of Widom type and the frequencies $\{\omega_j\}$ be independent. Then on the compact $J(E)$ there exists  unique shift invariant measure $dJ=d S^{-1}J S$. Moreover,  the generalized Abel map $\pi: J(E)\to \Gamma^*$ is essentially one-to-one, i.e.: $|\pi^{-1}(\Theta) |_{dJ}=0$ and on the set of the full measure, where $\pi$ is invertible, 
\begin{equation}\label{9part2}
d\pi(J)=d\alpha, \quad \alpha\in\Gamma^*\setminus\Theta.
\end{equation}

\end{theorem}

\begin{proof} Let $dJ$ be a shift invariant probability measure on the compact $J(E)$. Recall that $\pi$ is continuos. So, we can define $d\pi(J)$ in the standard way
\begin{equation}\label{10part2}
\int_O d\pi(J)=\int_{\pi^{-1}(O)}dJ
\end{equation}
for an open set $O\subset\Gamma^*$. 

According to Theorem \ref{th44} and Theorem \ref{thspectral10}, for a given $J\in J(E)$ we have a functional model related to a space $H^2(\alpha)$, $\alpha=\pi(J)$, with an evident corollary:
$\pi(S^{-1}J S)=\mu^{-1}\pi(J)$. This means that the measure $d\pi(J)$ is invariant with respect to 
shift by $\mu^{-1}$ in $\Gamma^*$:
\begin{eqnarray*}
\int_O d\,\mu^{-1}\pi(J)&=&\int_O d\pi(S^{-1}JS)=\int_{S^{-1}\pi^{-1}(O)S^{-1}}d J\\
&=&\int_{S^{-1}\pi^{-1}(O)S}dS^{-1}JS=\int_{\pi^{-1}(O)} dJ=\int_O d\pi(J).
\end{eqnarray*}

Let us point out that the frequencies are independent, and therefore
\begin{equation}\label{11part2}
\text{\rm clos}_{n\in \bbZ}\{\mu^n\alpha\}=\Gamma^*, \quad \text{for any}\ \alpha\in\Gamma^*.
\end{equation}
The independence of frequencies also implies immediately that there is only one $\mu$-shift invariant measure on $\Gamma^*$.
Thus $d\pi(J)=d\alpha$ is the unique Haar measure on $\Gamma^*$.

Then by Theorem \ref{th62} and the definition \eqref{10part2} $|\pi^{-1}(\Theta)|_{dJ}=0$.
On the remaining set $\pi^{-1}(\Gamma^*\setminus \Theta)$ we have $H^2_J=\hat H^2(\pi(J))=\check H^2(\pi(J))$,
that is, the map $\pi$ is one-to-one. Therefore $dJ$ is defined uniquely by \eqref{9part2} on the set 
$J(E)\setminus\pi^{-1}(\Theta)$ of full measure.

\end{proof}

\begin{corollary}
The shift invariant measure $dJ$ is ergodic, that is,  $\{J\}_{J(E)}$ is an ergodic family of Jacobi matrices.
\end{corollary}
\begin{proof}
By \eqref{9part2}.
\end{proof}

\section{Finite band approximation and \\ shift invariant  measure on $D(E)$}
\subsection{Finite band approximation}
We introduce a family of combs $\Pi_n$ with a finite number of slits of the form
\begin{equation}\label{fba1}
\Pi_n=\left\{z=w-\frac{i}{n}: w\in \Pi,\ \Im w>\frac 1 n\right\}.
\end{equation}
Let $\theta_n$ be the conformal mapping of the upper half-plane to $\Pi_n$ with the same normalization
$$
\theta_n(b_0)=-\pi, \quad \theta_n(a_0)=0,\quad \theta_n(\infty)=\infty.
$$
The set $E_n:=\theta_n^{-1}([-\pi,0])$ is a finite union of intervals. The covering mapping from $\bbD/\Gamma_n$ to $\Omega^n=\bar\bbC\setminus E_n$ is denoted by $\fz_n$.

Let $\cD_n$ be the connected component in the open set
$$
\{\zeta\in\bbD: |b(\zeta)|<e^{-\frac 1 n}\}
$$
which contains origin. Then, for the uniformization map $\phi_n:\bbD\to \cD_n$, $\phi_n(0)=0$, 
$\phi_n'(0)>0$, we have
\begin{equation}\label{0fba1}
b_n(\z)=e^{\frac 1 n}b(\phi_n(\z)),
\end{equation}
where $b_n=e^{i\theta_n}\circ \fz_n$ is the complex Green function of $\Omega^n$.

For a reader convenience, we formulate the standard relation between Hardy spaces in $\Omega$ and $\Omega^n$  in the following lemma.

\begin{lemma}\label{lfba1}
The superposition $f\circ\phi_n$ generates  
the map $\phi_n^*:\Gamma^*\to\Gamma_n^*$, moreover
\begin{equation}\label{fba2}
f\circ\phi_n\in H^2(\phi_n^*(\alpha)), \quad \| f\circ \phi_n\|_{H^2(\phi_n^*(\alpha))}\le \| f\|_{H^2(\alpha)}
\end{equation}
\end{lemma}
\begin{proof}
For an arbitrary closed curve $\gamma\in\pi_1(\Omega^n)$ we assign the closed curve $\phi_n(\gamma)\in\pi_1(\Omega)$ and this generates the map $\phi_n^*$.

If $u(\zeta)$ is the best harmonic majorant of $|f(\z)|^2$ then  $u(\phi_n(\z))$ is a harmonic majorant of $f\circ\phi_n$. If $u_n(\z)$ is the best majorant for this function we have $u_n(0)\le u(\phi_n(0))=u(0)$, which means \eqref{fba2}
\end{proof}

\begin{remark}
$\phi_n^*(\mu)=\mu_n^*$, where $\mu_n^*$ is the character of $b_n$.
\end{remark}
\begin{remark}\label{remfba3}
According to the definition of the domain $\Omega^n$ the critical points of the Green function in it correspond to the critical point $c_j$'s in $\Omega$, that is, if $\z_{c_j}\in \cD_n$ and $\fz(\z_{c_j})=c_j$, then $b'_n(\phi_n^{-1}(\z_{c_j}))=0$.
\end{remark}
\begin{remark}
Since $\Omega^n$ is finitely connected the hat- and check-Hardy spaces in it always coincide.
\end{remark}

We have the following important relation for the Widom functions. 
\begin{lemma}\label{lfba2}
Let $\Delta_n$ and $\Delta$ be the Widom functions in $\Omega^n$ and $\Omega$ respectively.
Then
\begin{equation}\label{fba3}
\Delta(\z)=\lim_{n\to\infty}\Delta_n(\phi_n^{-1}(\z)).
\end{equation}
\end{lemma}

\begin{proof}
First of all we note that $\cup_{n\ge 1} \cD_n=\bbD$, $\cD_n\subset\cD_{n+1}$, that is, starting from a certain $n$, $\phi^{-1}(\z)$ is well defined. Thus the sequence $\Delta_n(\phi_n^{-1}(\z))$ is well defined on an arbitrary compact set in $\bbD$. Recall that $|\Delta_n(\z)|\le 1$. 

We use  compactness of the family and consider an arbitrary convergent subsequence to define 
$$
f(\z)=\lim_{k\to\infty}\Delta_{n_k}(\phi_{n_k}^{-1}(\z)).
$$
According to  Remark \ref{remfba3}, $f(\z_{c_j})=0$, that is, $f=\Delta g$, where $\|g\|_{H^\infty}\le 1$. Evidently, $\lim_{n\to\infty}\sum h_j^{(n)}=\sum h_j$, where $h_j^{(n)}$ are the heights of the slits in $\Pi_n$. Therefore
$$
\lim_{n\to\infty}\Delta_n(0)=\Delta(0),
$$
that is, $g(0)=1$, thus $g=1$  identically and \eqref{fba3} is proved.
\end{proof}

\begin{theorem}
\label{fb}
 For the given finite band approximation
\begin{equation}\label{fba5}
\lim_{n\to\infty} k^{\phi_n^*(\alpha)}(0)=\hat k^\alpha(0).
\end{equation}
\end{theorem}

\begin{proof}
First of all by Lemma \ref{lfba1}
\begin{equation}\label{fba6}
\hat k^{\alpha}(0)\le k^{\phi_n^*(\alpha)}(0), \quad \forall \alpha\, \forall n.
\end{equation}

Let $\nu_n^*\in\Gamma_n^*$ be the character of the function $\Delta_n$. For each $n$ we fix a character $\nu_n\in\Gamma^*$ such that $\phi_n^*(\nu_n)=\nu_n^*$. We point out  that by Lemma \ref{lfba2}
$\lim_{n\to\infty}\nu_n=\nu$.
Since
$$
\check k^{\alpha^{-1}\nu_n}(0)\le \hat  k^{\alpha^{-1}\nu_n}(0)\le k^{\phi^*_n(\alpha)^{-1}\nu_n^*}(0)
$$
by the duality
$$
\frac{\Delta(0)^2}{\hat  k^{\alpha\nu_n^{-1}\nu}(0)}\le \frac{\Delta_n(0)^2}{k^{\phi^*_n(\alpha)}(0)}
$$
We can pass to the limit, having in mind that $\hat k^\alpha(0)$ is upper semi-continuous,
$$
\frac{\Delta(0)^2}{\hat  k^{\alpha}(0)}\le
\liminf_{n\to\infty}\frac{\Delta(0)^2}{\hat  k^{\alpha\nu_n^{-1}\nu}(0)}\le \liminf_{n\to\infty}
\frac{\Delta_n(0)^2}{k^{\phi^*_n(\alpha)}(0)}
=\frac{\Delta(0)^2}{\limsup_{n\to\infty}k^{\phi^*_n(\alpha)}(0)}.
$$
In combination with \eqref{fba6} we get \eqref{fba5}.
\end{proof}

\subsection{Shift invariant measure on the divisors}

Recall that the set $J(E)$ is homeomorphic to $D(E)$ (both are provided with corresponding weak topologies). Let $d\chi$ be  the measure on $D(E)$ corresponding to $dJ$  on $J(E)$.

\begin{theorem}
\label{thsid}
Let $\omega_k(z)$ be the harmonic measure of the set $E_k:=E\cap[b_k,a_0]$ evaluated at $z\in\Omega$. Then for an open set
\begin{equation}\label{sid1}
O=\{D: x_{i_1}\in I_{j_1}:=(a'_{j_1},b'_{j_1}),\dots, x_{i_\l}\in I_{j_\l}:= (a'_{j_\l},b'_{j_\l})\}
\end{equation}
where $[a'_{j_m},b'_{j_m}]\subset (a_{j_m},b_{j_m})$, and a fixed collections of $\{\epsilon_{j_m}\}$,
$1\le m\le \l$, we have
\begin{equation}\label{sid2}
\chi(O)=\frac 1 {2^\l}\left |\begin{matrix} 
\omega_{j_1}(b'_{j_1})-\omega_{j_1}(a'_{j_1})&\hdots\ &\omega_{j_1}(b'_{j_\l})-\omega_{j_1}(a'_{j_\l})\\
\vdots& &\vdots\\
\omega_{j_\l}(b'_{j_1})-\omega_{j_\l}(a'_{j_1})& \hdots&\omega_{j_\l}(b'_{j_\l})-
\omega_{j_\l}(a'_{j_\l})\\
\end{matrix}
\right|
\end{equation}

\end{theorem}

\begin{lemma}\label{l22sid}
If $E=E_N$ is a finite union of intervals then \eqref{sid2} holds.
\end{lemma}

\begin{proof} In this case the Haar measure is of the form $d\alpha=\prod_{j=1}^N dm(\alpha_j)$, where $dm$ is the Lebesgue measure on $\bbT$.
As it is well known for  the Abel map, which was defined in the general case  in \eqref{abm2}, we have
$$
{dm(\alpha_j)}=\frac 1 2 \sum^N_{k=1} \omega'_j(x_k){dx_k}
$$
Therfore
\begin{equation*}
d\alpha=dm(\alpha_1)\dots dm(\alpha_N)=\frac 1{2^N}\left |\begin{matrix} 
\omega'_{1}(x_1)&\hdots\ &\omega'_{1}(x_{N})\\
\vdots& &\vdots\\
\omega'_{N}(x_{1})& \hdots&\omega'_{N}(x_{N})
\end{matrix}\right| d x_1\dots dx_N.
\end{equation*}
We integrate over $O=O_N$ and use the fact that
\begin{equation*}
\omega_k(b_j)-\omega_k(a_j)=\delta_{k,j}.
\end{equation*}

\end{proof}

Thus the main point is to justify the passage to the limit. To this end we prove two preliminary lemmas.
The first one deals with a
uniform estimation of the reproducing kernel from below. 

\begin{lemma}
\label{l23sid}
Let $e^D(z)$ be the normalized reproducing kernel, written by means of the divisor (see \eqref{abm2})
\begin{equation}\label{sid3}
e^D(z)=\prod_{k\ge 1}b_{x_k}^{\frac{1+\epsilon_k}2}\sqrt{\prod_{k\ge 1}\frac{(z-x_k)b_{c_k}}{(z-c_k)b_{x_k}}}.
\end{equation}
Then for a rectangle
\begin{equation}\label{}
\Omega_j=\Omega_j(\delta)=\{z=x+iy: -\delta +a'_{j}\le x\le \delta+b'_j,\ |y|\le \delta \}\subset \Omega
\end{equation}
there is a uniform estimate 
\begin{equation}\label{sid4}
|e^D(z)|\ge C(\Omega_j)|z-x_j|.
\end{equation}

\end{lemma}

\begin{proof} First of all we note that for $x_k<a_j$ and $a_j<\Re z< b_j$ we have
$$
\frac{|z-x_k|}{|z-c_k|}\ge \frac{|z-b_k|}{|z-a_k|}.
$$ 
Therefore
$$
\prod_{x_k<a_j}\frac{|z-x_k|}{|z-c_k|}\ge \prod_{b_k<a_j}\frac{|z-b_k|}{|z-a_k|}\ge \frac{|z-a_j|}{|z-b_0|},
$$
and similarly,
$$
\prod_{x_k>b_j}\frac{|z-x_k|}{|z-c_k|}\ge \frac{|z-b_j|}{|z-a_0|}.
$$
For the product of $b_{x_k}$ we use the Widom condition. Indeed, 
there exists a large constant $M$ such that
\begin{equation}\label{sid5}
\frac 1 M G(z,z_0)\le G(z)
\end{equation}
for $|z|=r$ and $z\in \partial\Omega_j(2\delta)\subset \Omega$
uniformly on all $z_0\in \Omega_j(\delta)$.
Then \eqref{sid5}  holds in  
the region $\Omega_j(r,2\delta):=\Omega\cap\{z:|z|\le r\}\setminus\Omega_j(2\delta)$. Thus
$$
\sum_{k\not=j} G(x_k,z_0)\le M \sum_{k\not=j}G(x_k)\le M \sum_{k\ge 1} G(c_k),
$$
and so
$$
\prod_{k\not = j }|b_{x_k}(\fz^{-1}(z_0))|\ge \Delta(0)^M, \quad z_0\in \Omega_j(\delta).
$$
\end{proof}

We need a certain lifting of $\Omega_j$ into $\mathbb{D}$. Notice that lacunas $(a_j, b_j)$ correspond to certain arcs $A_j$ (of circles orthogonal to $\mathbb{T}$), and these arcs are parts of the boundary of the fundamental domain $F(\Gamma)$ of the Fuchsian group $\Gamma$ acting on the disc. We choose the connected component of $\fz^{-1}(\Omega_j)$ having a non-empty intersection with corresponding arcs $A_j$ on the boundary of $F(\Gamma)$. We call these components $K_j$. Unlike the whole pre-image $\fz^{-1}(\Omega_j)$, the sets $K_j$ are compactly contained in the unit disc, in particular,
\begin{equation}
\label{coj}
1-|\z| \ge \tau_j>0,\, \z\in K_j\,.
\end{equation}

\begin{lemma}
\label{l24sid}
For  $\Omega_j\subset\Omega$, there is a uniform estimate
\begin{equation}\label{sid6}
|\hat e^\alpha(\z)-e^{\phi_n^*(\alpha)}(\z)|^2\le C(\Omega_j)\left |1-\frac{\hat e^\alpha(0)}{e^{\phi_n^*(\alpha)}(0)} \right |,
\quad \z\in K_j.
\end{equation}

\end{lemma}

\begin{proof}
Recall that
$$
\|\hat e^\alpha-e^{\phi_n^*(\alpha)}\|^2=2-2\frac{\hat e^\alpha(0)}{e^{\phi_n^*(\alpha)}(0)}.
$$
Since both vectors are in $H^2$, we have
$$
|\hat e^\alpha(\z)-e^{\phi_n^*}(\z)|^2=\frac 1{1-|\z|^2}\|\hat e^\alpha-e^{\phi_n^*}\|^2.
$$
An application of \eqref{coj} finishes the proof.
\end{proof}

\begin{proof}[Proof of Theorem \ref{thsid}]
We fix a small positive $\delta$. By Lemma \ref{l23sid} there is a constant $c_1(\delta)>0$ such that

\begin{equation}\label{sid7}
|e^D(z)|\ge c_1(\delta), \quad z\in\cup_{s=1}^\l\partial\Omega_{j_s}(\delta)
\end{equation}
for every function of the form \eqref{sid3} with zeros $x_{i_s}\in I_{i_s}$, $1\le s\le \l$.
By Lemma \ref{l24sid} there is a constant $c_2(\delta)>0$ such that

\begin{equation}\label{sid8}
|\hat e^\alpha(\z)-e^{\phi_n^*(\alpha)}(\z)|\le c_2(\delta)\left |1-\frac{\hat e^\alpha(0)}{e^{\phi_n^*(\alpha)}(0)} \right |^{1/2},
\quad \z\in \cup_{s=1}^\l K_{j_s},
\end{equation}
for all $\alpha\in \Gamma^*$.

Let us fix 
\begin{equation}\label{sid9}
\epsilon\le \frac{1}{10}\frac{c_1(\delta)}{c_2(\delta)}
\end{equation}
and a small positive $\eta$. Since by Theorem \ref{fb}
$$
\lim_{n\to\infty}\int_{\Gamma^*}\left |1-\frac{\hat e^\alpha(0)}{e^{\phi_n^*(\alpha)}(0)} \right |^{1/2}d\alpha=0
$$
there exists $n=n(\delta,\eta)\ge 1/\delta$ such that

\begin{equation}\label{sid10}
\int_{\Gamma^*}\left |1-\frac{\hat e^\alpha(0)}{e^{\phi_n^*(\alpha)}(0)} \right |^{1/2}d\alpha
\le\epsilon\eta.
\end{equation}
Let $\Theta_n=\{\alpha\in\Gamma^*: \left |1-\frac{\hat e^\alpha(0)}{e^{\phi_n^*(\alpha)}(0)} \right |^{1/2}\ge \epsilon\}$. Then \eqref{sid10} implies  that its Haar measure is less than $\eta$, $|\Theta_n|\le \eta$. In other words, by \eqref{sid8},
\begin{equation}\label{sid11}
|\hat e^\alpha(\z)-e^{\phi_n^*(\alpha)}(\z)|\le c_2(\delta)\epsilon\le\frac 1{10} c_1(\delta), 
\quad \z\in \cup_{s=1}^\l K_{j_s},
\end{equation}
holds for all $\alpha\in\Gamma^*\setminus\Theta$.

\medskip

\noindent In what follows we denote by $Z^{-1}$ the branch of $\fz^{-1}$ that maps $\Omega_j$ onto $K_j$. We already considered the coverings $\fz_n:\mathbb{D}\rightarrow \Omega^n=\mathbb{C}\setminus E_n$. For conformal maps $\phi_n, \theta_n$ and covering maps $\theta_n$ we saw the following relationship
$$
\theta_n\circ \fz_n = \theta\circ\fz\circ \phi_n -\frac{i}{n}\,.
$$
There is one more useful conformal map. It maps the upper half-plane $\mathbb{C}_+$  onto $\{z\in \mathbb{C}_+: G(z) > \frac{1}{n}\}$ preserving $b_0, a_0, \infty$. Call it $\psi_n$. It plunges the lacunas of $\Omega^n$  into the lacunas of $\Omega$. It is easy to see that $\psi_n$ can be analytically continued across the lacunas of $\Omega^n$. It is also easy to see that for a fixed $\delta>0$ and all sufficiently large $n$ all $\Omega_{j_s}, s=1,\dots, \l,$ are in $\Omega^n$, and that $\psi_n$ are well defined on these sets. Moreover, $\psi_n\rightarrow \text{id}$ on $\Omega_{j_s}(2\delta), s=1,\dots, \l,$ when $n\rightarrow\infty$.
Along with the last display formula, we have another obvious one
$$
\theta_n = \theta\circ\psi_n-\frac{i}{n}\,.
$$
Inspecting these two formulae we get $\theta\circ\psi_n\circ \fz_n = \theta\circ\fz\circ\phi_n$. If we denote by $K_j^n:=\phi_n^{-1}(K_j)$ (these sets are defined defined for all large $n$), we will get now that 
$$
\psi_n \circ \fz_n = \fz \circ \phi_n
$$
on $\cup_{s=1}^\l K_{j_s}^n$.

We can consider analytic function $Z_n^{-1}$  given by $\phi_n^{-1}\circ Z^{-1}\circ \psi_n$ defined on $\cup_{s=1}^\l \Omega_{j_s}(2\delta)$. By previous formulae it is a branch of $\fz_n^{-1}$. Moreover, we know that $Z_n^{-1}\rightarrow Z^{-1}$ on compact $\cup_{s=1}^\l \Omega_{j_s}(2\delta)$, because $\psi_n, \phi_n$ converge to identities on corresponding compacts.

\medskip

  Without lost of generality  we assume that $\epsilon_{j_1}=\dots=\epsilon_{j_\l}=1$. Then (as in Section 6, we denote by $|\,\cdot\,|$ the Haar measure $d\alpha$ on $\Gamma^*$)
\begin{equation}\label{sid12}
\chi(O)=\left| \{\alpha\in \Gamma^*: \hat e^{\alpha}(Z^{-1}(x_{j_s}))=0, \ x_{j_s}\in I_{j_s},\ 1\le s\le \l\} \right|
\end{equation}

\noindent Recall that all zeros of 
$e^{\phi_n^*(\alpha)}(\fz^{-1}(z))$ are real (and not more than one in a gap $(a_j,b_j)$).

Let $I_{j_s}(\delta)=\Omega_{j_s}(\delta)\cap\ \bbR$, $1\le s\le \l$.
Comparing \eqref{sid7} and \eqref{sid11}, by
Rouch\'e's Theorem we can continue \eqref{sid12} with  the following chain of estimates
\begin{eqnarray*}
\chi(O)
&\ge&\left| \{\alpha\in \Gamma^*\setminus \Theta_n: \hat e^{\alpha}(Z^{-1}(x_{j_s}))=0, \ x_{j_s}\in I_{j_s},\ 1\le s\le \l\} \right|\\
&\ge&\left| \{\alpha\in \Gamma^*\setminus \Theta_n: e^{\phi_n^*(\alpha)}(Z^{-1}(x_{j_s}))=0, \ x_{j_s}\in I_{j_s}(\frac12\delta),\ 1\le s\le \l\} \right|\\
&\ge&\left| \{\alpha\in \Gamma^*\setminus \Theta_n: e^{\phi_n^*(\alpha)}(Z_n^{-1}(x_{j_s}))=0, \ x_{j_s}\in I_{j_s}(\delta),\ 1\le s\le \l\} \right|\\
&\ge&\left|\{\alpha \in \Gamma^*: e^{\phi_n^*(\alpha)}(Z_n^{-1}(x_{j_s}))=0, \ x_{j_s}\in I_{j_s}(\delta),\ 1\le s\le \l\} \right| -\eta\,.
\end{eqnarray*}
But the last expression is exactly the Haar measure of the open set
\begin{equation*}
O_n(\delta)=\{D_n: x_{i_1}\in I_{j_1}(\delta),\dots, x_{i_\l}\in I_{j_\l}(\delta)\}
\end{equation*}
in the $n$-th finite band approximation $E_n$ of the set $E$, see the previous section. Thus,  we can apply Lemma \ref{l22sid} to this set. As the result we get
\begin{eqnarray*}
&\chi(O)&\ge-\eta+\chi^{(n)}(O_n(\delta))\ge\\
&+&
\frac 1 {2^\l}\left |\begin{matrix} 
\omega^{(n)}_{j_1}(b'_{j_1}+\delta)-\omega^{(n)}_{j_1}(a'_{j_1}-\delta)&\hdots\ &\omega^{(n)}_{j_1}(b'_{j_\l}+\delta)-\omega^{(n)}_{j_1}(a'_{j_\l}-\delta)\\
\vdots& &\vdots\\
\omega^{(n)}_{j_\l}(b'_{j_1}+\delta)-\omega^{(n)}_{j_\l}(a'_{j_1}-\delta)& \hdots&\omega^{(n)}_{j_\l}(b'_{j_\l}+\delta)-
\omega^{(n)}_{j_\l}(a'_{j_\l}-\delta)\\
\end{matrix}
\right|-\eta
\end{eqnarray*}
where $\omega^{(n)}_k(z)$ is the harmonic measure related to the $n$-th finite band approximation of $E$ (that is of $\Omega^n$).

Summarizing all that we conclude that fixing $\delta>0$ and $\eta>0$ we will have for all sufficiently large $n$
$$
\chi(O)\ge \chi^{(n)}(O_n(\delta))-\eta\,,
$$
where $\chi^{(n)}$ denotes the determinant above. Tend $n$ to infinity. Then harmonic measures in the determinant will become harmonic measures with respect to  our limit domain $\Omega=\mathbb{C}\setminus E$. Then subsequently tending $\eta$ to zero and $\delta$ to zero we obtain inequality from below of Theorem \ref{thsid}. The inequality from above is obtained in a totally similar manner.


\end{proof}

\begin{remark}
We will use only the estimate from below, but note that the formula \eqref{sid2} can be used in exact asymptotics for zeros of orthogonal polynomials.
\end{remark}

\section{$J\in J(E)$ is not  almost periodic \\ as soon as DCT fails}

In this section we finish a proof of Theorem \ref{thmain}. We assume that the conditions $(i-iii)$ hold and DCT fails, that is, $\Theta$ is not empty.

\begin{lemma}
There exists $J_0\in J(E)$, which is not (uniformly) almost periodic.
\end{lemma}
\begin{proof}
Let $\alpha\in\Theta$. Let $\hat J=J(\hat H^2(\alpha))$ and $\check J=J(\check H^2(\alpha))$, see Theorem \ref{th44}. Note that $\hat J\not=\check J$. Let $\beta\in\Gamma^*\setminus\Theta$. 
By \eqref{11part2} we chose a sequence $\{m_n\}$ such that $\alpha\mu^{-m_n}\to\beta$. 
Assume now that both $\hat J$ and $\check J$ are almost periodic. Then there is a subsequence, for which we keep the notation $\{m_n\}$, such that
$$
\lim_{n\to\infty}\|S^{-m_{n}}\hat J S^{m_{n}}- J_1\|=0, \quad 
\lim_{n\to\infty}\|S^{-m_{n}}\check J S^{m_{n}}- J_2\|=0.
$$
Since uniform convergence implies pointwise convergence and $\pi$ is continuous we have $\pi(J_1)=\pi(J_2)=\beta$.
But  $\beta\in\Gamma^*\setminus \Theta$, therefore  $J_1= J_2=\pi^{-1}(\beta)$. That is,
$$
\lim_{n\to\infty}\|S^{-m_{n}}\hat J S^{m_{n}}-
S^{-m_{n}}\check J S^{m_{n}}\|=0,
$$
but
$$
\|S^{-m_{n}}\hat J S^{m_{n}}-
S^{-m_{n}}\check J S^{m_{n}}\|=\|\hat J-\check J\|\not =0.
$$
Thus, at least one of the matrices $\hat J$, $\check J$ is not almost periodic. This one we denote by $J_0$.
\end{proof}

As always we denote the orbit by $\text{orb} (J):=\{S^{-n}JS^n\}_{n\in \mathbb{Z}}$. We need now a small and very well known

\begin{lemma}
\label{orb}
Let $J$ be almost periodic,  and let $J_0$ be in a weak closure of $\text{orb}(J)$. Then $J_0$ belongs to the operator closure of $\text{orb}(J)$. In particular, 
$J_0$ is also almost periodic.
\end{lemma}

\begin{proof}
Let the subsequence $S^{-m_n}JS^{m_n}$ converges to $J_0$ weakly. Then, using that $J$ is almost periodic we can find  a subsequence $\{m_{n_k}\}$ such that
$S^{-m_{n_k}}J S^{m_{n_k}}$ converges to some $J_1$ in the operator topology. Obviously $J_1=J_0$. Then the whole orbit of $J_0$  is contained  in the operator closure of the orbit of $J$. The latter is compact, so $\text{orb}(J_0)$ is pre-compact, and we are done.
\end{proof}

\begin{lemma}
\label{ae}
 Almost every $J\in J(E)$ (with respect to the shift invariant measure $dJ$) is not almost periodic.
\end{lemma}

\begin{proof}
The idea is to use that due to Theorem \ref{thsid} open subsets of $J(E)$ have strictly positive measure $dJ$. Let $J_0$ be the Jacobi matrix from the previous lemma. Let $V_n$ be a system of vicinities 
of $J_0$ such that
$$
V_1\supset V_2\supset\dots \supset V_n\supset\dots\ \text{and}\ \cap_{n\ge 1}V_n=J_0,
$$
or, equivalently, using the homeomorphism between $J(E)$ and $D(E)$, and denoting the corresponding sets by $O_n$, we can write
$$
O_1\supset O_2\supset\dots \supset O_n\supset\dots\ \text{and}\ \cap_{n\ge 1}O_n=D_0,
$$
where $O_n$ are vicinities of the divisor $D_0$, which corresponds to $J_0$.

Let 
$$
\Sigma_n= \pi\bigcup_{m\in \mathbb{Z}}S^{-m}V_nS^m\setminus \Theta\,.
$$
By Theorem \ref{thsid} we have $|V_n|_{dJ}=\chi(O_n)>0$. Then,  by ergodicity,
$$
|\Sigma_n|_{d\alpha}=|\pi^{-1}(\Sigma_n)|_{dJ}=1.
$$
Thus we have the set of the full measure
$$
\Sigma:=\cap_{n\ge 1}\Sigma_n.
$$

By definition, $J_0$ is in the weak closure of $\text{orb}(J(\beta))$ for any $\beta\in \Sigma$. By Lemma \eqref{orb} matrix $J_0$ is also almost periodic, this is the contradiction, and Lemma \ref{ae} is proved.


\end{proof}

\begin{proof}[Finishing the proof of Theorem \ref{thmain}]
Let $\beta$ belongs to the set $\Sigma\subset \Gamma^*\setminus \Theta$ from the previous lemma.
 Let $J\in J(E)$, $\pi(J)=\alpha$, and let $\alpha\mu^{-m_n}\to\beta$.
Then $J(\beta)$ is in the weak closure of the orbit of $J$.  Assume that this  $J$ is almost periodic. Then, by Lemma \ref{orb}, it is almost periodic. But  $J(\beta)$ is not almost periodic for the given $\beta$
by Lemma \ref{ae}. This is the contradiction, and Theorem \ref{thmain} is completely proved.

\end{proof}

\section{Appendix: Resolvent representation and Chri\-stoffel-Darboux identity for the transfer matrix}

For $\cK=\hat H^2(\alpha)\ominus \check H^2(\alpha)$ we define
$T:\cK\oplus\{\check e_0\}\to \{\hat e_{-1}\}\oplus \cK$ by
\begin{equation}\label{transfer6}
T=P_{\{\hat e_{-1}\}\oplus \cK}\,\fz\cdot\,| \cK\oplus\{\check e_0\}.
\end{equation}
We define $\hat \cE$ and $\check \cE$ acting in $\bbC^2$ by
\begin{equation}\label{transfer7}
\check \cE\begin{bmatrix} c_{-1}\\ c_0\end{bmatrix}=c_{-1}\check p_0\check e_{-1}+c_0\check e_0 
,\quad
\hat \cE\begin{bmatrix} c_{-1}\\ c_0\end{bmatrix}=c_{-1}\hat p_0\hat e_{-1}+c_0\hat e_0. 
\end{equation}
\begin{proposition} For $T$, $\hat \cE$, $\check \cE$ defined in \eqref{transfer6}, \eqref{transfer7},
the transfer matrix \eqref{transfer4} is of the form
\begin{equation}\label{transfer8}
\fA(z)=\begin{bmatrix}
0&0\\0&\langle \hat e_0,\check e_0\rangle
\end{bmatrix}+ j
\check \cE^*P_{\cK\oplus\{\check e_0\}}(T-z P_\cK)^{-1}P_{\{\hat e_{-1}\}\oplus \cK}\hat \cE.
\end{equation}

\end{proposition}

\begin{proof}
With necessity the vector
$(T-z_0 P_\cK)^{-1}\hat e_{-1}$ is of the form
$$
x\oplus c \check e_0=\frac{\hat e_{-1}-A \check e_{-1}-B\check e_{0}}{\fz-z_0}\oplus c \check e_0.
$$
Since $x\in H^2$ we have
\begin{equation}\label{transfer9}
\hat e_{-1}(\z_0)=\begin{bmatrix} \check e_{-1}(\z_0)& \check e_{0}(\z_0)
\end{bmatrix}
\begin{bmatrix} A\\ B
\end{bmatrix}
\end{equation}

Since $x\in \hat H^2(\alpha)\ominus \check H^2(\alpha)$ the dual vector $\tilde x(\z)$, such that  $b(\z)\tilde x(\z)=\Delta(\z) x(\bar\z)$ for $\z\in\bbT$, belongs to  
$\hat H^2(\alpha^{-1}\mu^{-1}\nu)\ominus \check H^2(\alpha^{-1}\mu^{-1}\nu)$. Thus, in notations 
\eqref{hsd0} we have
\begin{equation}\label{transfer10}
\widetilde{\hat e}_{0} (\z_0)=\begin{bmatrix} \widetilde{\check e}_{0}(\z_0)& \widetilde{\check e}_{-1}(\z_0)
\end{bmatrix}
\begin{bmatrix} A\\ B
\end{bmatrix}
\end{equation}

In this case, indeed,
\begin{eqnarray*}
(T-z_0P_\cK) x&=&P_\cK(\fz-z_0)x+P_{\{\hat e_{-1}\}}\fz x\\
&=& P_\cK (\hat e_{-1}-A \check e_{-1}-B \check e_{0})+\hat e_{-1}\langle x, \hat e_0\rangle
\hat p_0\\
&=&-A P_\cK \check e_{-1}+\hat p_0\frac{(b\hat e_{-1})(0)-A(b\check e_{-1})(0)}{(b\fz)(0)\hat e_0(0)}\hat e_{-1}
\end{eqnarray*}
and
\begin{equation*}\label{darbu1}
T\check e_0=P_{{\{\hat e_{-1}\}\oplus \cK}} \check p_0 \check e_{-1}=\check p_0 P_{\cK} 
\check e_{-1}+ \check p_0\frac{(b\check e_{-1})(0)}{(b\hat e_{-1})(0)}\hat e_{-1}
\end{equation*}
That is, for $c=A/\check p_0$ we have
$$
(T-zP_\cK)(x+c \check e_0)=\left(A\frac{(b\check e_{-1})(0)}{(b\hat e_{-1})(0)}
+\frac{(b\hat e_{-1})(0)-A(b\check e_{-1})(0)}{(b \hat e_{-1})(0)}\right)\hat e_{-1}=
\hat e_{-1}
$$
Thus
\begin{equation}\label{0transfer1}
(T-zP_\cK)^{-1} \hat e_{-1}=
\frac{\hat e_{-1}-A \check e_{-1}-B\check e_{0}}{\fz-z_0}\oplus \frac{A}{\check p_0} \check e_0,
\end{equation}

Similarly we get
\begin{equation}\label{0transfer2}
(T-zP_\cK)^{-1} P_\cK \hat e_{0}=
\frac{\hat e_{0}-C \check e_{-1}-D\check e_{0}}{\fz-z_0}\oplus \frac{C}{\check p_0} \check e_0,
\end{equation}
where
\begin{equation}\label{transfer11}
\begin{bmatrix}\hat e_{0}(\z_0)\\
\widetilde{\hat e}_{-1} (\z_0)
\end{bmatrix}
=\begin{bmatrix} \check e_{-1}(\z_0)& \check e_{0}(\z_0)\\
 \widetilde{\check e}_{0}(\z_0)& \widetilde{\check e}_{-1}(\z_0)
\end{bmatrix}
\begin{bmatrix} C\\ D
\end{bmatrix}.
\end{equation}

We see that all four entries $A,B,C,D$ given in \eqref{transfer9}, \eqref{transfer10},
\eqref{transfer11} actually form the transfer matrix $\fA$.

Using \eqref{0transfer1}, \eqref{0transfer2} and the reproducing properties of the vectors
$\check e_{-1}, \check e_0$ we get \eqref{transfer8}.

\end{proof}

\begin{lemma} The following Christoffel-Darboux type identity holds
\begin{equation}\label{0transfer12}
\displaystyle
\frac{\fA^*(z)j\fA(z)-j}{z-\bar z}=\hat \cE^*P_{\{\hat e_{-1}\}\oplus \cK}(T^*-\bar z P_\cK)^{-1}
P_\cK(T-z P_\cK)^{-1}P_{\{\hat e_{-1}\}\oplus \cK}\hat \cE.
\end{equation}

\end{lemma}

\begin{proof} We start with an easy identity
\begin{equation}\label{transfer12}
P_\cK T-T^*P_\cK=\check \cE \begin{bmatrix}0&1\\-1&0\end{bmatrix}\check\cE^*.
\end{equation}
Indeed, for $f=x+ c\check e_{0}$ we have 
\begin{eqnarray*}
\langle P_\cK T (x+c\check e_{0}),x+c\check e_{0}\rangle&=&
\langle \fz (x+c\check e_{0}),x\rangle=
\langle \fz x,x\rangle+c \langle \check p_0 \check e_{-1},x\rangle\\
&=&\langle \fz x,x\rangle+\check p_0
\langle f,\check e_{0}\rangle \langle \check e_{-1},f\rangle.
\end{eqnarray*}
We can rewrite  \eqref{transfer12} in a more sophisticated form
\begin{equation}\label{0darbu2}
(z-\bar z) P_{\cK}=\check \cE \begin{bmatrix}0&1\\-1&0\end{bmatrix}\check\cE^*
- P_\cK (T-zP_k)+(T^*-\bar z P_k)P_\cK.
\end{equation}
Now we multiply \eqref{0darbu2} by $P_{\cK\oplus\{\check e_0\}}(T-z P_\cK)^{-1}P_{\{\hat e_{-1}\}\oplus \cK}\hat \cE$ from the right and by the conjugated expression from the left. We use
\eqref{0transfer1}, \eqref{0transfer2} and the reproducing properties of $\hat e_{-1}, \hat e_0$ to obtain \eqref{0transfer12}.

\end{proof}

\bibliographystyle{amsplain}

\bigskip

Department of Mathematics, Michigan State University, East Lansing,
MI 48824, USA

\textit{E-mail address}: {volberg@math.msu.edu}

\smallskip

{Institute for Analysis, Johannes Kepler University Linz,
A-4040 Linz, Austria} 

\textit{E-mail address}:
{Petro.Yudytskiy@jku.at}
\end{document}